\documentclass[11pt]{article}
\usepackage{amsmath,amsthm,amssymb}
\usepackage{verbatim}
\usepackage[margin=1in]{geometry}
\usepackage{graphicx}
\usepackage{subfig,wrapfig}
\usepackage{url}
\usepackage{hyperref}

\newif\ifabstract
\abstracttrue
\abstractfalse
\newif\iffull
\ifabstract \fullfalse \else \fulltrue \fi

\usepackage{times}

\ifabstract

\newlength\abovesectionskip
\newlength\belowsectionskip
\def\sectionfont{\normalfont\Large\bfseries}
\newlength\abovesubsectionskip
\newlength\belowsubsectionskip
\def\subsectionfont{\normalfont\large\bfseries}
\newlength\abovesubsubsectionskip
\newlength\belowsubsubsectionskip
\def\subsubsectionfont{\normalfont\normalsize\bfseries}
\newlength\aboveparagraphskip
\newlength\belowparagraphskip
\def\paragraphfont{\normalfont\normalsize\bfseries}
\makeatletter
\def\section{\@startsection{section}{1}{\z@}{-\abovesectionskip}%
               {\belowsectionskip}{\sectionfont}}
\def\subsection{\@startsection{subsection}{2}{\z@}{-\abovesubsectionskip}%
                  {\belowsubsectionskip}{\subsectionfont}}
\def\subsubsection{\@startsection{subsubsection}{3}{\z@}%
                     {-\abovesubsubsectionskip}{\belowsubsubsectionskip}%
                     {\subsubsectionfont}}
\def\paragraph{\@startsection{paragraph}{4}{\z@}{-\aboveparagraphskip}%
                 {\belowparagraphskip}{\paragraphfont}}
\makeatother

{\makeatletter
 \gdef\subsection{\@ifnextchar*\subsection@star\subsection@normal}
 \gdef\subsection@normal#1{\refstepcounter{subsection}%
            \paragraph{\thesubsection\hbox{~~}#1.}}
 \gdef\subsection@star*#1{\paragraph{#1.}}}

\makeatletter
\def\GrabProofArgument[#1]{ #1: \egroup\ignorespaces}
\def\proof{\noindent\textbf\bgroup Proof%
           \@ifnextchar[{\GrabProofArgument}{: \egroup\ignorespaces}}
\def\endproof{\hspace*{\fill}$\Box$\medskip}
\makeatother

\fi

\newtheorem{thm}{Theorem}
\newtheorem{lem}[thm]{Lemma}

\theoremstyle{definition}

\newcommand{\term}[1]{\emph{#1}}
\newcommand{\tri}[2]{\mathcal T_{#1}(#2)}
\newcommand{\sector}[3]{\mathcal S_{#1}(#2,#3)}
\newcommand{\kitesweep}[2]{\mathcal{KS}_{#1}(#2)}
\newcommand{\chain}[2]{\mathcal C_{#1}(#2)}
\newcommand{\join}[4]{{#1}_{(#2)}\wedge {}_{(#4)}{#3}}

\newcommand{\freee}{\operatorname{free}}

\def\captionfont{\small}
\def\captionlabelfont{\bf\small}
{\makeatletter
 \global\let\plainfont@makecaption=\@makecaption
 \long\gdef\@makecaption#1#2{%
   \plainfont@makecaption{\captionlabelfont #1}{\captionfont #2}}}

\def\andlinebreak{\end{tabular}\linebreak\begin{tabular}[t]{c}}

\title{Hinged Dissections Exist}
\author{%
  Timothy G. Abbott%
    \thanks{MIT Computer Science and Artificial Intelligence Laboratory,
      32 Vassar Street, Cambridge, MA 02139, USA,
      \protect\url{{tabbott,edemaine,mdemaine}@mit.edu}}
    \thanks{Partially supported by an NSF Graduate Research Fellowship
            and an MIT-Akamai Presidential Fellowship.}
\and
  Zachary Abel%
    \thanks{Department of Mathematics, Harvard University,
      1 Oxford Street, Cambridge, MA 02138, USA.
      \protect\url{{zabel,kominers}@fas.harvard.edu}}
    \thanks{Corresponding author.}
\and
  David Charlton%
    \thanks{Department of Computer Science, Boston University,
      111 Cummington Street, Boston, MA 02215, USA.
      \protect\url{charlton@cs.bu.edu}}
\andlinebreak
  Erik D. Demaine\footnotemark[1]
  \thanks{Partially supported by NSF CAREER award CCF-0347776,
          DOE grant DE-FG02-04ER25647, and AFOSR grant FA9550-07-1-0538.}
\and
  Martin L. Demaine\footnotemark[1]
\and
  Scott D. Kominers\footnotemark[3]
}
\date{}

\begin{document}
\maketitle

\begin{abstract}
  We prove that any finite collection of polygons of equal area has a
  common hinged dissection.  That is, for any such collection
  of polygons there exists a chain of polygons
  hinged at vertices that can be folded in the plane continuously
  without self-intersection to form any polygon
  in the collection.  This result settles the open problem about
  the existence of hinged dissections between pairs of polygons that
  goes back implicitly to 1864 and has been studied extensively in the
  past ten years.  Our result generalizes and indeed builds upon the
  result from 1814 that polygons have common dissections (without
  hinges).  We also extend our common dissection result to edge-hinged
  dissections of solid 3D polyhedra that have a common (unhinged)
  dissection, as determined by Dehn's 1900 solution to Hilbert's Third
  Problem.  Our proofs are constructive, giving explicit algorithms in all cases.  For a constant number of planar polygons, both the
  number of pieces and running time required by our construction are
  pseudopolynomial.  This bound is the best possible, even for unhinged
  dissections.  Hinged dissections have possible applications to
  reconfigurable robotics, programmable matter, and nanomanufacturing.
\end{abstract}

%%% Title page for Submission
\setcounter{page}0
\thispagestyle{empty}
\clearpage

\section{Introduction}

\begin{wrapfigure}{r}{2in}
  \centering
  \vspace{-6ex}
  \includegraphics[width=\linewidth]{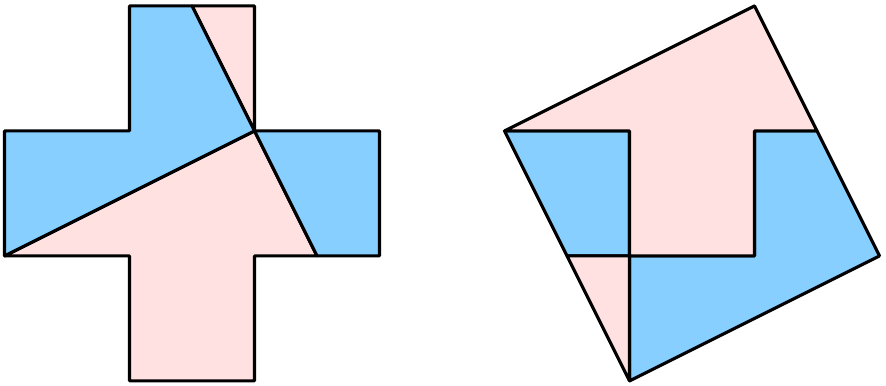}
  \iffull
    \caption{4-piece dissection of Greek cross to square from 1890
             \cite{Lemon-1890}.}
  \else
    \caption{1890 dissection of Greek cross to square \cite{Lemon-1890}.}
  \fi
  \label{fig:greek cross to square}
\end{wrapfigure}
Around 1808, Wallace asked whether every two polygons of the same area
have a common \emph{dissection}, that is, whether any two equal-area
polygons can be cut into a finite set of congruent polygonal pieces
\cite[p.~222]{Frederickson-1997}.  Figure~\ref{fig:greek cross to
  square} shows a simple example.  Lowry \cite{Lowry-1814} published
the first solution to Wallace's problem in 1814, although Wallace may
have also had a solution at the time; he published one in 1831
\cite{Wallace-1831}.  Shortly thereafter, Bolyai \cite{Bolyai-1832}
and Gerwien \cite{Gerwien-1833} rediscovered the result, whence this
result is sometimes known as the Bolyai-Gerwien Theorem.

By contrast, Dehn \cite{Dehn-1900} proved in 1900 that not all polyhedra of the
same volume have a common dissection, solving Hilbert's Third Problem
posed in the same year \cite{Dehn-1900}.  Sydler \cite{Sydler-1965}
showed that Dehn's invariant in fact characterizes 3D dissectability.

Lowry's 2D dissection construction,
as described by Frederickson \cite{Frederickson-1997},
is particularly elegant and uses a pseudopolynomial number of pieces.%
\footnote{In a geometric context, \emph{pseudopolynomial} means polynomial
  in the combinatorial complexity ($n$) and the dimensions of the integer grid
  on which the input is drawn.  Although the construction does not require
  the vertices to have rational coordinates, a pseudopolynomial analysis
  makes sense only in this case.}
A pseudopolynomial bound is the best possible in the worst case:
dissecting a polygon of diameter $x > 1$ into a polygon of diameter
$1$ (for example, a long skinny triangle into an equilateral triangle)
requires at least $x$ pieces.  With this worst-case result in hand,
attention has turned to optimal dissections using the fewest pieces
possible for the two given polygons. This problem has been studied
extensively for centuries in the mathematics literature
\cite{Ozanam-1778,Cohn-1975,Frederickson-1997} and the puzzle
literature
\cite{Panckoucke-1749,Lemon-1890,Madachy-1979-dissections,Lindgren-1972},
and more recently in the computational geometry literature
\cite{Czyzowicz-Kranakis-Urrutia-1999,Kranakis-Krizanc-Urrutia-2000,
  Akiyama-Nakamura-Nozaki-Ozawa-Sakai-2003}.

\emph{Hinged dissections} are dissections with an additional
constraint: the polygonal pieces must be hinged together at vertices
into a connected assembly.  The first published hinged dissection
appeared in 1864, illustrating Euclid's Proposition I.47
\cite{Kelland-1864}; see \cite[pp.~4--5]{Frederickson-2002}.  The most
famous hinged dissection is Dudeney's 1902 hinged dissection
\cite{Dudeney-1902}; see Figure~\ref{fig:dudeney}.  This surprising
construction inspired many to investigate hinged dissections; see, for
example, Frederickson's book on the topic \cite{Frederickson-2002}.

However, the fundamental problem of general hinged dissection has
remained open \cite{topp,cgcolumn}: do every two polygons of the same
area have a common hinged dissection?  This problem has been attacked
in the computational geometry literature
\cite{Akiyama-Nakamura-1998,polyforms,Eppstein-2001,polyforms3d}, but
has only been solved in special cases.  For example, all polygons made
from edge-to-edge gluings of $n$ identical subpolygons (such as
polyominoes) have been shown to have a common hinged dissection
\cite{polyforms}.  Perhaps most intriguingly, Eppstein
\cite{Eppstein-2001} showed that the problem of finding a hinged
dissection of any two triangles of equal area is just as hard as the
general problem.

\begin{figure*}[htbp]
  \centering
  \includegraphics[width=\linewidth]{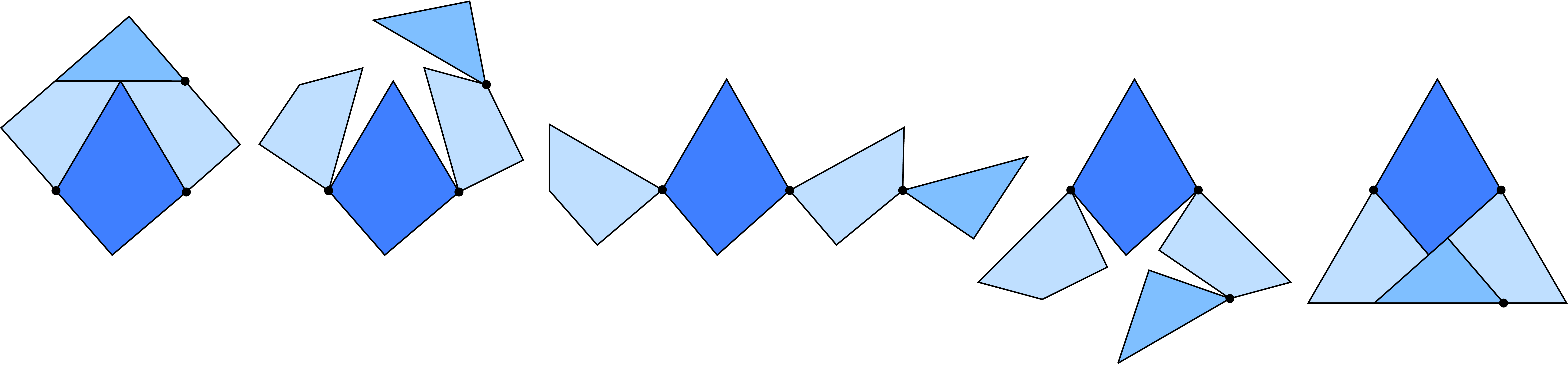}
  \caption{Dudeney's 1902 hinged dissection of a square into a triangle
           \protect\cite{Dudeney-1902}.}
  \label{fig:dudeney}
\end{figure*}

Hinged dissections are particularly exciting from the perspectives of
reconfigurable robotics, programmable matter, and nanomanufacturing.
Recent progress has enabled chemists to build millimeter-scale
``self-working'' 2D hinged dissections such as Dudeney's
\cite{Mao-Thallidi-Wolfe-Whitesides-Whitesides-2002}.
An analog for 3D hinged dissections may enable the building of a complex 3D structure
out of a chain of units; see \cite{Griffith-2004-PhD} for one such approach.
%If the process is programmable, w
We could even envision an object that can
re-assemble itself into different 3D structures on demand \cite{polyforms3d}.
This approach contrasts existing approaches to reconfigurable robotics
(see, for example, \cite{Rus-Butler-Kotay-Vona-2002}), where units must
reconfigure by attaching and detaching from each other through a complicated
mechanism.

\paragraph{Our results.}

We settle the hinged dissection open problem, first formally posed in a
CCCG 1999 paper \cite{polyforms} but implicit back to 1864 \cite{Kelland-1864}
and 1902 \cite{Dudeney-1902}.
Specifically, Section~\ref{sec:equi} proves a universality result:
any two polygons of the same area have a common hinged dissection.
In fact, our result is stronger, building a single hinged dissection
that can fold into any finite set of desired polygons of the same area.
The analogous multipolygon result for (unhinged) dissections is
obvious---simply overlay the pairwise dissections---but no such general
combination technique is known for hinged dissections.
Indeed, the lack of such a transitivity construction has been the main
challenge in constructing general hinged dissections.

Our construction starts from an arbitrary (unhinged) dissection, such
as Lowry's \cite{Lowry-1814}.  We show that any dissection of a finite
set of polygons can be subdivided and hinged so that the resulting
hinged dissection folds into all of the original polygons.  We give a
method of subdividing pieces of a hinged figure which effectively
allows us to ``unhinge'' a portion of the figure and ``re-attach'' it
at an alternate location.
% we show that two groups of pieces connected by a single
%hinge can be ``re-attached'' at any desired point in addition to
%the existing hinge, simply by subdividing the pieces.
This construction allows us to ``move'' pieces and hinges around
arbitrarily, at the cost of extra pieces. Therefore, we are able to
hinge any dissection.

This na\"ive construction easily leads to an exponential number of
pieces, but we show in Section~\ref{sec:pseudopoly} that a more
careful execution of Lowry's dissection \cite{Lowry-1814} attains a
pseudopolynomial number of pieces for a constant number of target
polygons.  As mentioned above, such a bound is essentially best
possible, even for unhinged dissections.  
% (Admittedly, however, the exponents of our polynomials leave room
% for improvement.)
This more efficient construction requires significantly more complex
gadgets for simultaneously moving several groups of pieces at roughly
the same cost as moving a single piece, and relies on specific
properties of Lowry's dissection.

We also solve another open problem concerning the precise model of
hinged dissections.  In perhaps the most natural model of hinged dissections,
pieces cannot properly overlap during the folding motion from one
configuration to another.  However, all theoretical work concerning hinged
dissections \cite{Akiyama-Nakamura-1998,polyforms,Eppstein-2001,polyforms3d}
has only been able to analyze the ``wobbly hinged'' model
\cite{Frederickson-2002}, where pieces may intersect during the motion.
Is there a difference between these two models?
Again this problem was first formally posed at CCCG 1999 \cite{polyforms}.
We prove in Section~\ref{sec:motion} that any wobbly hinged dissection
can be subdivided to enable continuous motions without piece intersection,
at the cost of increasing the combinatorial complexity of the
hinged dissection by only a constant factor.
This result builds on expansive motions from the Carpenter's Rule Theorem
\cite{Connelly-Demaine-Rote-2003,Streinu-2005} combined with the theory
of slender adornments from SoCG 2006 \cite{planarshapes}.

The following theorem summarizes our results for 2D figures:

\begin{thm} \label{thm:main2D} Any finite set of polygons of equal
  area have a common hinged dissection which can fold continuously
  without intersection between the polygons.  For a constant number of
  target polygons with vertices drawn on a rational grid, the number
  of pieces is pseudopolynomial, as is the algorithm to compute the
  common hinged dissection.
\end{thm}

Finally, we generalize our results to 3D in Section~\ref{sec:3D}.
As mentioned above, not all 3D polyhedra have a common dissection
even without hinges.  Our techniques generalize to show that
hinged dissections exist whenever dissections do:

\begin{thm} \label{thm:main3D}
  If two 3D polyhedra have the same volume and Dehn invariant,
  then they have a common hinged dissection.
\end{thm}

\section{Terminology}\label{sec:terminology}

\begin{wrapfigure}{r}{3in}
  \centering
  \vspace*{-4ex}
  \includegraphics[scale=.65]{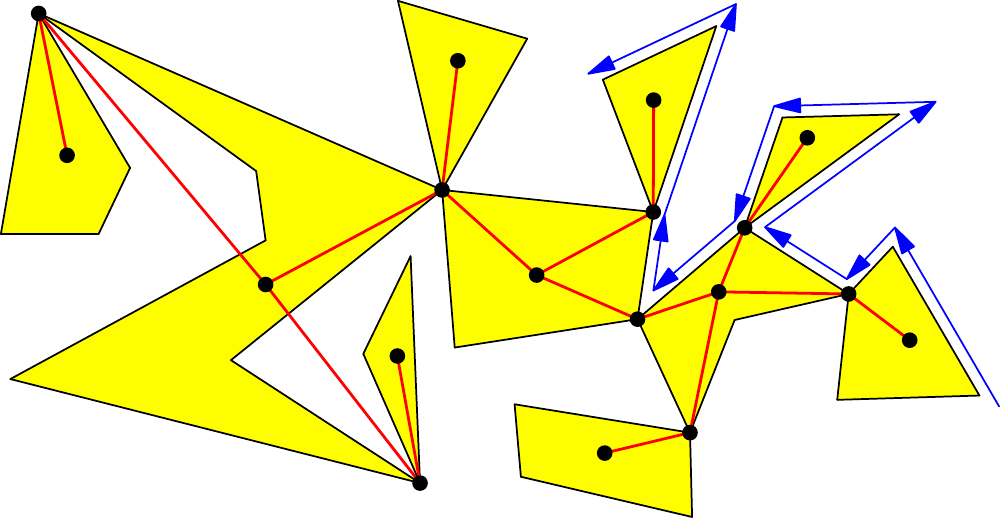}
  \caption{A hinged figure, its incidence graph (red), and part of is
    boundary path.}
  \label{fig:inc_graph}
\end{wrapfigure}

% Hinged Figures
A \term{hinged figure} $F$ is a finite collection of simple, oriented
polygons (the \term{links}) hinged together at rotatable joints at the
links' vertices so that the resulting figure is connected, together
with a fixed cyclic order of links around each hinge. (Note that a
hinge might exist at a $180^\circ$ angle of a link, but this hinge is
still considered a vertex of the link.)
A \term{configuration} of a hinged figure $F$ is an embedding of $F$'s
links into the plane so that the links' interiors are disjoint and so
that each hinge's cyclic link order is maintained.

The \term{incidence graph} of a hinged figure is the graph that has a
vertex corresponding to every link and every hinge, such that two
nodes are connected by an edge if one represents a link and the other
represents a hinge on that link. See Figure~\ref{fig:inc_graph}.
A hinged figure is \term{tree-like} if the incidence graph is a tree,
and it is \term{chain-like} if the incidence graph is an open chain.

The \term{boundary} $\partial A$ of a hinged figure $A$ is the
oriented path (or collection of paths) along the edges of the links
traversed in depth-first order, as illustrated in Figure
\ref{fig:inc_graph}. For a tree- or chain-like figure, the boundary
consists of a single path incorporating all edges of the links. Note
that the boundary path will trace each hinge point multiple times, but
we distinguish these as different boundary points.
% However, by \term{boundary point} we shall refer to the point's
% position along the path, not directly to its position along the
% geometry of the figure or the configuration, and so the boundary
% points $a$, $b$, and $c$ in Figure \ref{fig:boundary} are all
% distinct.

% Hinged Figure Refinement
For two hinged figures $A$ and $B$, we say that $B$ is a
\term{refinement} of $A$, and write $B\prec A$, if $A$ can be obtained
from $B$ by gluing together portions of $B$'s boundary, i.e.,
by adding hinges between pieces of $B$ which may effectively glue
together shared edges of pieces in~$B$.
Intuitively, one could obtain $B$ from $A$ by cutting the pieces in $A$
and breaking some of the hinges.
%(although it is then more difficult to guarantee
%that all configurations of $A$ induce configurations of~$B$).
The gluing in the definition gives rise to an imposed configuration of $B$
for every configuration of~$A$.
The property of refinement is transitive;
that is, if $C\prec B$ and $B\prec A$, then $C\prec A$.
This transitivity of refinement plays a central role in the arguments below.

\section{Universal Hinged Dissection}
\label{sec:equi}

In this section, we show that any finite collection of equal-area
polygons has a common hinged dissection.
More precisely, we construct a hinged figure with a configuration in
the shape of every desired polygon; continuous motions
without intersection will be addressed in Section~\ref{sec:motion}.
%
%\begin{thm}
%  For any pair of polygons $P$ and $Q$ of equal area, there exists a
%  hinged figure $A$ that has two configurations whose silhouettes
%  are, respectively, $P$ and $Q$.
%\end{thm}
%
The proof is in three parts: effectively moving rooted subtrees,
effectively moving rooted pseudosubtrees, and arbitrarily rearranging
pseudotrees.

\subsection{Moving Rooted Subtrees}
Consider a tree-like hinged figure $F$. If there are two hinged
figures $A$ and $B$ with two distinguished boundary points
$a\in\partial A$ and $b\in\partial B$ so that $F$ is equivalent to the
hinged figure obtained by identifying points $a$ and $b$ to a single
hinge (denoted $F = \join{A}{a}{B}{b}$), then we say $A$ and $B$ are
each \term{rooted subtrees} of $F$. If another boundary point
$b'\in\partial B$ is chosen, then the new hinged figure $F' =
\join{A}{a}{B}{b'}$ is related to $F$ by a \term{rooted subtree
  movement}: $(A,a)$ is the subtree that has been moved.

Our goal is to accomplish this movement with hinged dissection.
We will achieve this goal by connecting pieces with
\term{chains} of isosceles triangles hinged at their base vertices. We
begin with a lemma concerning cutting isosceles
triangles from polygons, and then proceed to construct the required
dissection by cutting out chains from both $A$ (at the point $a$) and
$B$ (along the boundary from $b$ to $b'$).

For an angle $\alpha<90^\circ$ and a length $\ell$, denote by
$\tri{\alpha}{\ell}$ the isosceles triangle with base of length $\ell$
and base-angles $\alpha$. For a segment $PQ$, use the notation
$\tri{\alpha}{PQ}$ for the triangle $\tri{\alpha}{|PQ|}$ drawn with
base along segment $PQ$. Finally, for an angle $\beta$, point $P$, and
radius $r$, let $\sector{\beta}{P}{r}$ be a circular sector centered
at $P$ with angle $\beta$ and radius $r$.

\begin{lem}\label{lem:free-regions}
  For any simple polygon $V=V_1\ldots V_n$, there exist an angle
  $\beta$ and a radius $r$ small enough so that the triangles
  $\tri{\beta}{V_iV_{i+1}}$ constructed inward along the edges, as
  well as circular sectors $\sector{\beta}{V_i}{r}$ drawn inside $V$,
  are pairwise disjoint except at the vertices of $V$.  These
  triangles and sectors will be called the \term{free-regions} for
  their respective edges or vertices of $V$.
\end{lem}

\ifabstract
\begin{wrapfigure}{r}{3in}
\else
\begin{figure}
\fi
  \centering
  \vspace{-2ex}
  \includegraphics[scale=.65]{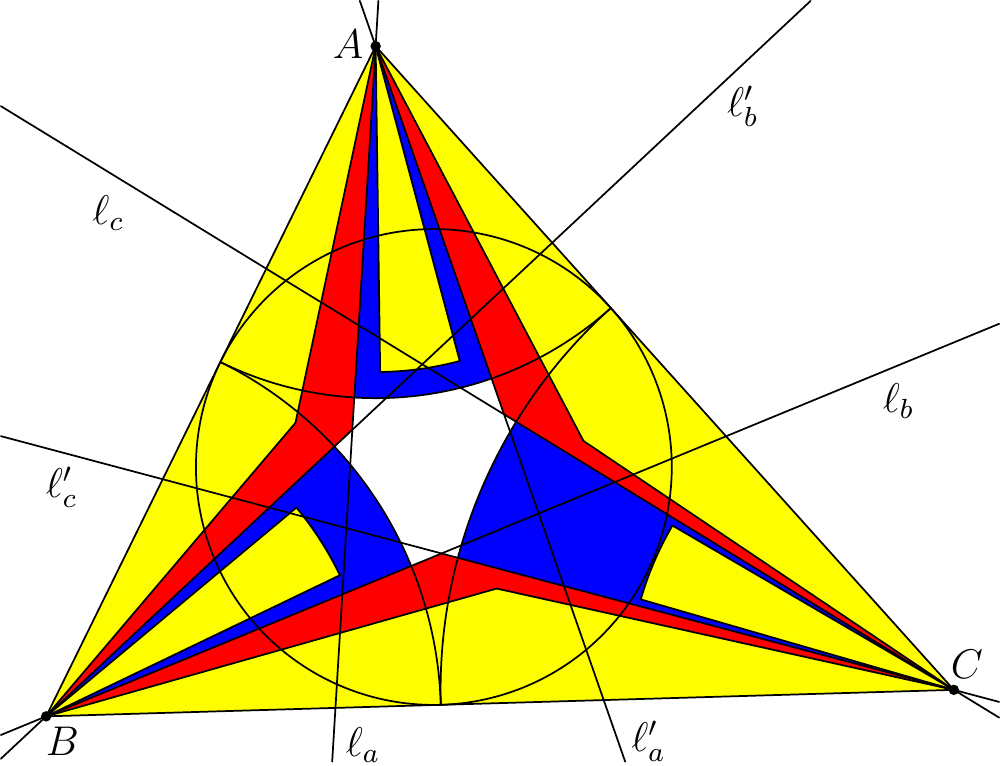}
  \caption{The free regions (lightly colored) in triangle $ABC$ are separated
    by the dark circles and the angle trisectors $\ell_a$, $\ell_a'$,
    etc.}
  \vspace{-3ex}
  \label{fig:freeregion}
\ifabstract
\end{wrapfigure}
\else
\end{figure}
\fi

\proof
  We first prove the result for triangles. For triangle $T=ABC$ with
  side lengths $a,b,c$, semiperimeter $s=\frac{1}{2}(a+b+c)$, and
  angles $\delta,\epsilon,\zeta$, choose $\beta_T <
  \frac{1}{3}\min\{\delta,\epsilon,\zeta\}$ and $r_T <
  \min\{s-a,s-b,s-c\}$. Then the triangles $\tri{\beta_T}{AB}$, etc.,
  and the sectors $\sector{\beta_T}{A}{r_T}$, etc., can be drawn in
  the triangle without overlap, as in Figure \ref{fig:freeregion}:
  Indeed, $\tri{\beta_T}{AB}$ is contained between $AB$ and the two
  trisectors $\ell_a$ and $\ell_b'$ (the region shown in red), sector
  $\sector{\beta_T}{A}{r_T}$ is contained in the sector
  $\sector{\frac{\delta}{3}}{A}{s-a}$ between trisectors $\ell_a$ and
  $\ell_a'$ (shown in green), etc., and these six regions are
  interior-disjoint.

  For the general case, first triangulate polygon $V=V_1\cdots V_n$ by
  $n-2$ diagonals. For each triangle $T=V_iV_jV_k$ in the
  triangulation, calculate $\beta_T$ and $r_T$ as above, and draw the
  free regions in $T$. Finally, as all resulting triangles and sectors
  are disjoint (except at vertices), choosing $\beta =
  \min_{T}\{\beta_T\}$ and $r=\min_T\{r_T\}$ suffices.
\endproof

For a sequence of positive lengths $\ell_1,\ldots,\ell_n$, we define
the \term{chain} $\chain{\alpha}{\ell_1,\ldots,\ell_n}=C$ to be the
hinged figure formed by hinging the $2n$ upward-pointing triangles
\begin{equation*}
  \tri{\alpha}{\ell_1},
  \tri{\alpha}{\ell_1}, \tri{\alpha}{\ell_2}, \tri{\alpha}{\ell_2}, \ldots,
  \tri{\alpha}{\ell_n}, \tri{\alpha}{\ell_n}
\end{equation*}
in order at their base vertices. The \term{initial point} $C_0$ and
\term{final point} $C_1$ of this chain are the unhinged vertices of
the first $\tri{\alpha}{s_1}$ and the last $\tri{\alpha}{s_n}$
respectively.

% \begin{thm}
%   Consider a treelike figure $B$ with two boundary points
%   $b,b'$. Let $\gamma$ be either of the two boundary paths from $b$
%   to $b'$. Then for any $r > 0$, it is possible to find a sequence
%   of lengths $s_1,\ldots,s_n$ with $s_i < r$ for all $1\le i\le n$
%   such that for all sufficiently small $\alpha > 0$ there exists a
%   rooted hinged figure $(C,c)$ with the following properties:
%   \begin{enumerate}
%   \item The hinged figure
%     \begin{equation*}F = (C,c)\wedge
%       (\chain{\alpha}{s_1,\ldots,s_n},\operatorname{final}(\chain{\alpha}{s_1,\ldots,s_n})),
%     \end{equation*}
%     where the chain is rooted at its initial or final point
%     (depending on the orientation of $\gamma$ around $B$) is a
%     refinement of $B$;
%   \item
%   \end{enumerate}
% \end{thm}

We may now state and prove the desired result of this section.

% Pick a single link $L$ and a hinge $h$ of $L$. If we were to cut $F$
% into two pieces by removing $L$ from $h$, the half containing $L$ is
% called the \term{rooted subtree} of $L$ and $h$, and the other half
% is the aptly named \term{other half}. The goal of this section is to
% effectively move a rooted subtree so that it links with the other
% half at a different boundary point.

\begin{thm}\label{thm:move_subtree}
  For any two tree-like figures $F$ and $F'$ related by the rooted
  subtree movement of $(A,a)$ from $(B,b)$ to $(B,b')$, there exists a
  common refinement $G\prec F$ and $G\prec F'$.
  Further, if $a$ lies on link $L_a\in A$, and a simple path $\gamma$
  along $\partial B$ is chosen from $b$ to $b'$, this refinement
  $G\prec F$ may be chosen so that only $L_a$ and links incident with
  $\gamma$ are refined.
\end{thm}

Note first that both $A$ and $B$ are tree-like, as they are subtrees of
tree-like figure $F$.  Note also that there are exactly two boundary
paths $\gamma$ from $b$ to $b'$ since $B$ is tree-like.

\begin{figure}
  \centering
  \subfloat[][]{\includegraphics[scale=.65]{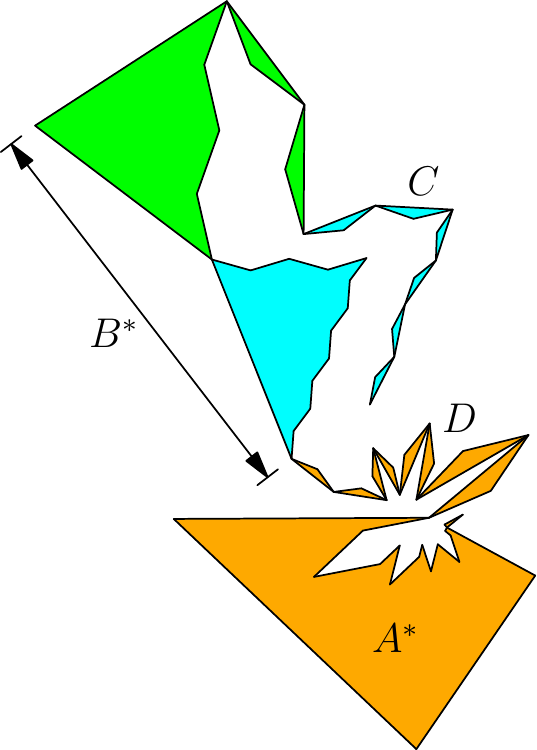}}\hfil
  \subfloat[][]{\includegraphics[scale=.65]{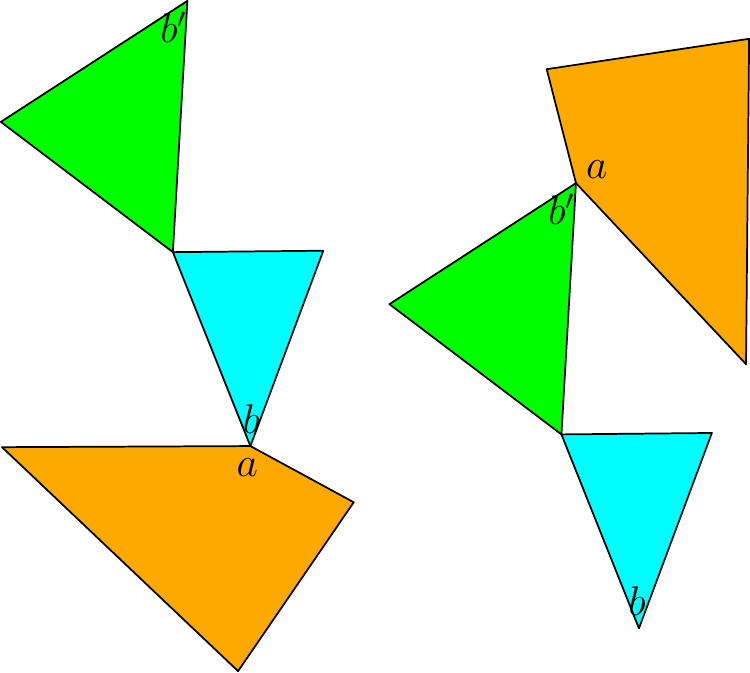}}\hfil
  \subfloat[][]{\includegraphics[scale=.65]{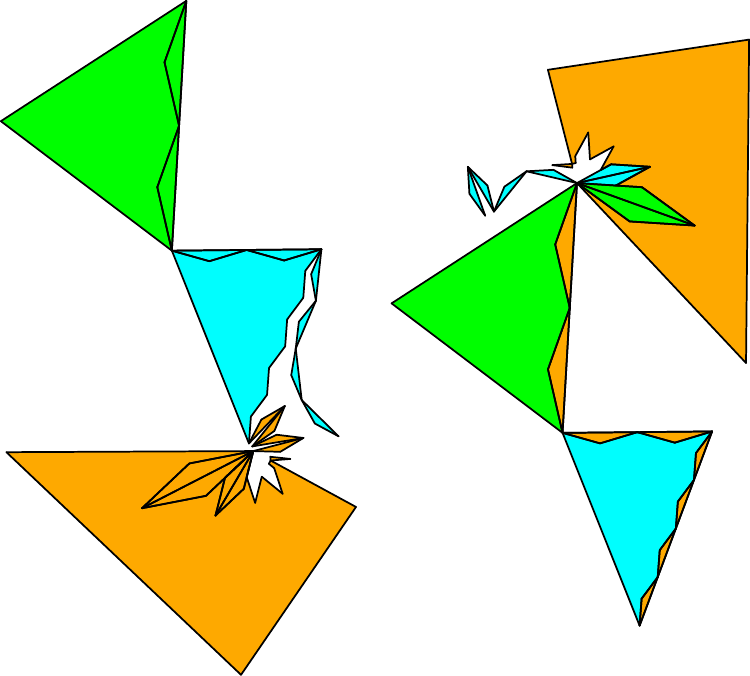}}
  \caption{Effectively moving a rooted subtree}
  \label{fig:movefixed}
\end{figure}

\medskip

\begin{proof}
  Without loss of generality, the diagram is oriented so that $\gamma$
  traces the boundary of $B$ counterclockwise from $b$. The
  construction is in two steps.

  In the first step, we cut a chain from the boundary of $\gamma$, as
  follows.  Let $r$ be the smallest free-region radius for all links
  touched by $\gamma$, and likewise let $\alpha$ be the smallest
  free-region angle.
  Path $\gamma$ is a polygonal path $P_0P_1\ldots P_t$ along the
  boundary of $B$, where $P_i$ are vertices of links with $P_0 = b$
  and $P_t = b'$. By refining this path into shorter segments as
  necessary, we may assume that each segment $P_{i-1}P_{i}$ has length
  $2\ell_i$ with $\ell_i\le r$.

  % Suppose edge $P_{i-1}P_{i}$ along $\gamma$ lies along link $L_i$
  % of $B$. Fore each $i$, choose an integer $m_i$ large enough so
  % that $|P_{i-1}P_{i}|/2m_i < r$,

  % and add flatt vertices to $L_i$ to partition $P_iP_{i+1}$ into
  % $2m_i$ equal segments of length $|P_iP_{i+1}|/2m_i$. This
  % partitions $\gamma$ into $2k = 2(m_1+\cdots+m_t)$ segments
  % $\ell_1,\ldots,\ell_{2k}$ such that $|\ell_{2j-1}| = |\ell_{2j}| =
  % s_j$ for each $1\le j\le k$.

  Choose an angle $\beta < \alpha/2t$. Next, cut out $2t$ isosceles
  triangles along $\gamma$: for each segment $P_iP_{i+1}\in\gamma$,
  cut two $\tri{\beta}{\ell_i}$ triangles. These triangles fit in the
  appropriate free-triangle for their link in $B$ by choice of
  $\beta$, so all of these triangles may indeed be removed without
  overlapping or disconnecting any of $B$'s links. Let $B^{*}$ be the
  hinged figure after these triangles have been removed, and let
  $C=\chain{\beta}{\ell_1,\ldots,\ell_t}$ be the chain formed by
  hinging these $2t$ cut-out triangles in order.  Finally, rehinge the
  pieces to form the figure $G_b = \join{B^{*}}{b'}{C}{C_1}$. See
  Figure \ref{fig:movefixed} for an illustration.

  The other step is to cut a chain away from $A$ at $a$. Draw $t$
  abutting rhombi $r_1,\ldots,r_t$ in link $L_a$ at point $a$ so that
  $r_i$ has a diagonal of length $\ell_i$ and an angle of $2\beta$;
  they are drawn in the order $r_1,\ldots,r_k$ clockwise around $a$ so
  that $r_i$ shares (part of) an edge with $r_{i+1}$ for $1\le i\le
  k-1$. Call this configuration of kites a \term{kite-sweep}
  $\kitesweep{\beta}{\ell_1,\ldots,\ell_n}$.
  Recall that $\beta$ was chosen so that $2t\beta < \alpha$ and that
  $\ell_i$ were chosen so that $\ell_i < r$ for each $1\le i\le r$, so
  this kite-sweep can fit within the free-sector of $L_a$ at $a$.
  Finally, cut out these $t$ rhombi in the form of $2t$
  $\beta$-triangles, rehinging them into
  $D=\chain{\beta}{\ell_1,\ldots,\ell_t}$.  Link $L_a$ is no longer a
  simple polygon, so simply cut away a small corner near $a$ and
  rehinge it as shown in Figure \ref{fig:movefixed}.  Let $A^*$ be the
  remaining hinged figure after $A$ has been thus mutilated.
  Finally, hinge all of $A$ back together in the form $G_a =
  \join{A^{*}}{a}{D}{d_1}$.

  The final result of our construction is the single hinged figure $G
  = \join{G_a}{D_0}{G_b}{b}$; I claim $G\prec F$ and $G\prec F'$. To
  see $G\prec F$, simply configure chains $C$ and $D$ so that each
  link assumes the spot from which it was cut from $F$; i.e., chain
  $C$ fills the triangular holes left along path $\gamma$, and chain
  $D$ fills the kite-sweep in $L_a$. See Figure
  \ref{fig:movefixed}(c), left.
  For the refinement $G\prec F'$, the chains simply switch roles:
  chain $D$ now fills in the gaps left along $\gamma$, and chain $C$
  fills the kite holes in $L_a$. See Figure \ref{fig:movefixed}(c),
  right.
%
%
  % Next, for each segment $CD$ along the polygonal path $\gamma$,
  % divide $CD$ into $2m$ segments of length $|CD|/2m$, where $m$ is
  % large enough so that $|CD|/2m < r$.
  %   %
  % If there are $2k$ total such segments (let's call these segments
  % $s_1,\ldots,s_{2k}$) along $\gamma$, choose an angle $\beta$ small
  % enough so that $\beta < \frac{\alpha}{2k}$. Also, for each link
  % $L$ sharing an edge with $\gamma$, consider $L$ to have vertices
  % added at all endpoints of incident edges $s_i$, and choose $\beta$
  % smaller than the free-isosceles angle for all such $L$.
%
  % We then cut out $2k$ isosceles triangles with base-angle $\beta$
  % and bases along $s_i$, which are disjoint by choice of $\beta$.
  % We also cut out $k$ abutting rhombi $r_i$ ($1\le i\le n$) with
  % angle $2\beta$ and diagonal of length $|s_{2i-1}| = |s_{2i}|$
  % counterclockwise around vertex $a$ of $L_a$, as shown in figure
  % \ref{fig:isos-tris}; by the choice of $\alpha$ and $\beta$, these
  % rhombi do indeed fit in the circular sector.
  %   %
  % Finally, to preserve the simple-polygon condition, we cut off a
  % tiny corner of the mutilated $L_a$ at $a$.
  %   %
  % Finally, we hinge the pieces as illustated in Figure
  % \ref{fig:isos-tris}, thus forming two thin \term{chains} of
  % isosceles triangles. As illustrated, this hinged figure has two
  % configurations corresponding to $F_1 = (A,a)\wedge(B,b_1)$ and
  % $F_2 = (A,a)\wedge(B,b_2)$, as desired.  In both configurations,
  % one of the chains fills in the space along $\gamma$ and the other
  % fills in the gap left in $L_a$.
\end{proof}

% The refinement formed here is not tree-like, but it can be made so
% by removing any hinge along either chain.  The rest of the proof
% below simply iterates this rooted subtree refinement, being careful
% that the procedure does indeed terminate.

\subsection{Moving Rooted Pseudosubtrees}

Now we increase the level of abstraction by allowing movement of
rooted subtrees in a hinged figure $F$ that already has a refinement
$G\prec F$.  We call $F$ the \term{pseudofigure} of $G$, and subtrees
of $F$ \term{pseudosubtrees} of $G$.

\begin{thm}\label{thm:move_pseudosubtree}
  Take tree-like figures $F$ and $F'$ related by the rooted-subtree
  movement of $(A,a)$ from $(B,b)$ to $(B,b')$ as in
  Theorem~\ref{thm:move_subtree}, and suppose $G\prec F$.  Then there exists
  a common refinement $H\prec G\prec F$ and $H\prec F'$. Further, if a
  path $\gamma$ from $b$ to $b'$ on $\partial B$ is chosen, then only
  links of $G$ incident with $\gamma$ are refined.
\end{thm}

In other words, this theorem allows the movement of a pseudo-subtree
of $G$. The construction below directly generalizes the method used in
Theorem~\ref{thm:move_subtree}.

\begin{figure}
  \centering
  \subfloat[][]{\includegraphics[scale=.65]{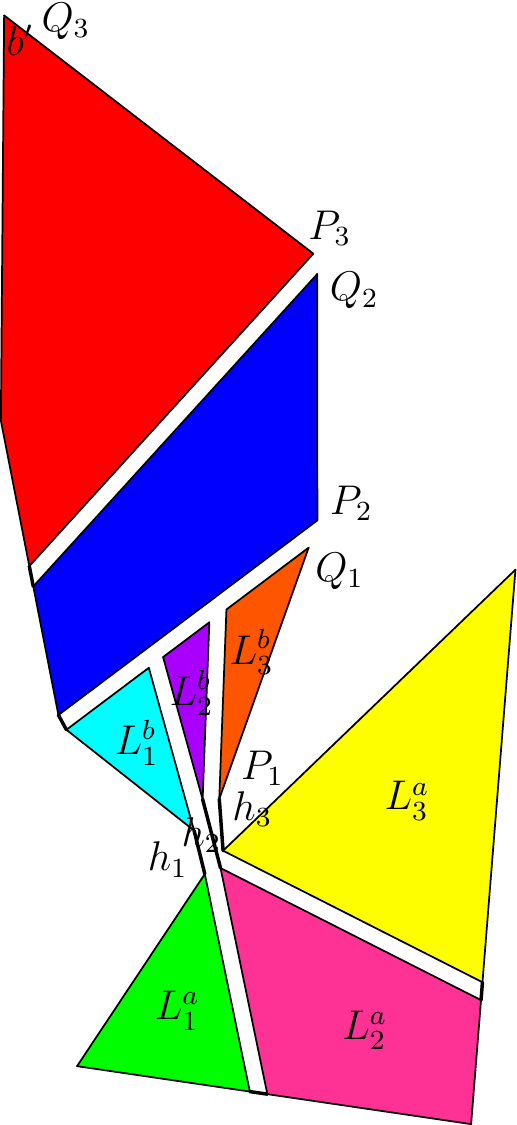}}
  \subfloat[][]{\includegraphics[scale=.65]{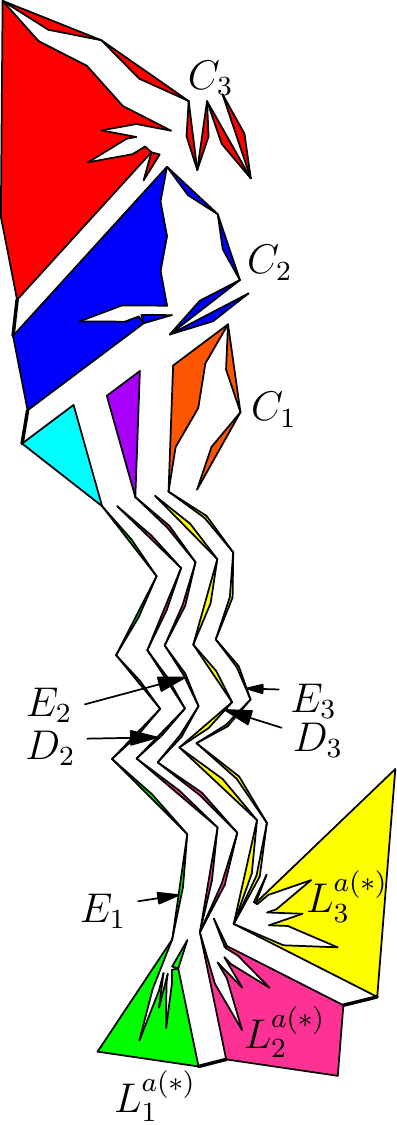}}
  \subfloat[][]{\includegraphics[scale=.65]{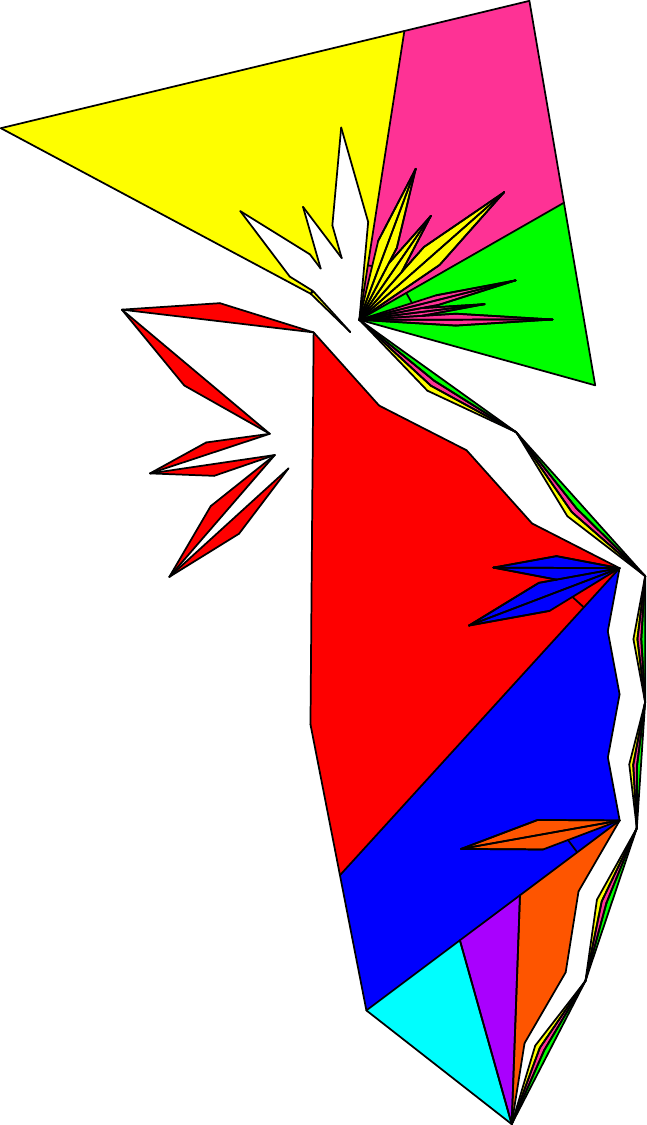}}
  \caption{Moving a rooted pseudosubtree.}
  \label{fig:movepseudodisconnected}
\end{figure}

\medskip

\begin{proof}
  We will write from $G$'s point of view, so features of $F$ will have
  the \emph{pseudo} prefix.
  Without loss of generality, suppose $\gamma$ winds counterclockwise
  around the pseudoboundary of $B$.

  Consider the behavior near the pseudohinge $h$ of $F$ corresponding
  to points $a$ and $b$; let $\{h_i\mid 1\le i\le n\}$ be the set of
  all the hinges of $G$ with the property that $h_i$ has two incident
  links $L_i^a$ and $L_i^b$ lying in $A$ and $B$ respectively. As
  links are defined to be simple polygons, these $2n$ links are
  distinct.
  We may suppose that these $2n$ links are the \emph{only} links of
  $G$ incident with pseudo-hinge $h$: the construction below is
  unchanged by the presence of more, extraneous links.
  Without loss of generality, we may assume that these links have been
  numbered so that they fall in the cyclic order
  $L_1^a,\ldots,L_{n-1}^a,L_n^a,L_n^b,L_{n-1}^b,\ldots,L_1^b$
  counterclockwise around $h$.

  Our goal is to mimic the two steps in the proof of
  Theorem~\ref{thm:move_subtree}, by effectively cutting a chain from $A$ at
  $a$ and cutting a chain from $B$ along $\gamma$.

  We begin by choosing the dimensions of the chain. First, the
  refinement $G\prec F$ induces an identification of some boundary
  points of $G$, and any point $p\in\partial G$ collocated with a
  vertex of any link in $G$ will itself be declared a (possibly flat)
  vertex of its link. We also declare $b'$ to be a vertex of its link,
  if it isn't already.
  Let $r$ and $\alpha$ be the smallest free-region radius and angle
  for any link in $G$ incident with $\gamma$.  Polygonal path $\gamma$
  consists of $t$ segments $P_iQ_i$, ($1\le i\le t$) from the boundary
  of $G$, where $P_1$ corresponds to $b$ and $Q_t$ corresponds to
  $b'$; as before, we may subdivide $\gamma$ as necessary so that
  $|P_iQ_i| = 2\ell_i \le 2r$ for each $1\le i\le n$.  We choose
  $\beta = \alpha/2t$.

  Note that boundary point $Q_i$ of $G$ does not necessarily equal
  $P_{i+1}$, but if they are unequal then both $Q_i$ and $P_{i+1}$ are
  vertices of their respective links; let $i_0,\ldots,i_s$ be the
  indices where $Q_{i_j-1}\neq P_{i_j}$, with $i_0 = 1$ and $i_s =
  t+1$.
%
  % Choose a radius $r$ \zref{small enough!}, and an angle $\beta$
  % \zref{small enough!}. Next, partition each segment $P_iQ_i$ into
  % $2m_i$ segments of length $|P_iQ_i|/2m_i < r$, giving a total of
  % $2k = 2(m_1+\cdots+m_t)$ segments $s_1,\ldots,s_{2k}$ of lengths
  % $|s_{2i-1}| = |s_{2i}| = \ell_i$ as before. For convenience, $q_j
  % = m_1+\cdots+m_{i_{j}-1}$ is half the number of segments along
  % $\gamma$ from $P_0$ to $P_{i_j}$.

  We begin by refining the links along $\gamma$ to imitate the first
  step in the construction of Theorem~\ref{thm:move_subtree}, i.e. to
  simulate cutting a $\chain{\beta}{\ell_1,\ldots,\ell_t}$ from
  $\gamma$ and linking it onto $b'$. We treat each portion
  $P_{i_j+1}Q_{i_{j+1}}$ of $\gamma$, corresponding to a contiguous
  path along $\partial G$, separately.
  For each $i_j\le k\le i_{j+1}-1$, cut two $\tri{\beta}{\ell_k}$
  triangles inward along $P_kQ_k$. Also, cut a kite-sweep
  $\kitesweep{\beta}{\ell_1,\ldots,\ell_{i_j-1}}$ from the free-sector
  at $P_{i_j}$, and then make the link simple by removing and
  rehinging a small corner as shown in Figure
  \ref{fig:movepseudodisconnected}(b).  Notice that, since $\beta$ is less
  than the free-region angle along each edge incident with $\gamma$,
  all of the removed isosceles triangles fit within this
  region. Likewise, the kite-sweep has total angle $2\beta\cdot
  (i_{j}-1)\le 2t\beta=\alpha$ and the largest kite has diagonal
  $\max\{\ell_1,\ldots,\ell_{i_j-1}\}\le r$, so the kite-sweep fits in
  the free-sector at $P_{i_j}$.
  We have now removed $2(i_{j+1}-1)$ triangles, namely two of each
  length $\ell_1,\ldots,\ell_{i_{j+1}-1}$, which we now rehinge into a
  chain $C_{j}=\chain{\beta}{\ell_1,\ldots,\ell_{i_{j+1}-1}}$ and
  attach to $G$ by hinging $C_j$'s final point to $Q_{i_{j+1}-1}$.

  To see that this construction refines $G$, note that each of the $s$
  chains may simply fill in the places from which they were cut. To
  see that this hinged figure can also serve the purpose that
  $\join{B^{*}}{b'}{C}{c_0}$ serves in Theorem~\ref{thm:move_subtree},
  note that each chain $C_j$ may fill in the kite-sweep cut at
  $P_{i_{j+1}}$ for $1\le j\le s-1$, while $C_j$ is the desired chain
  attached at $b'$ (Figure \ref{fig:movepseudodisconnected}(c)).

  Now we show how to refine $G$ around pseudohinge $h$. For each $1\le
  i\le n$, cut a $\kitesweep{i\beta/n}{\ell_1,\ldots,\ell_k}$
  kite-sweep in $L_i^a$ at $h_i$; the resulting non-simple link has
  two corners at $h_i$, so we cut off and rehinge the more
  counterclockwise of the two, calling the resulting link (without this
  small corner) $L_i^{a(*)}$.
  As before, by choice of $r$ and $\beta$, the $i$th kite-sweep can
  fit within the free-sector of $L_i^a$ at $h_i$. For $2\le i\le n$,
  cut each of the $2t$ triangles $\tri{i\beta/n}{\ell_j}$ removed from
  $L_i^a$ into two pieces: a triangle $\tri{(i-1)\beta/n}{\ell_j}$
  with the same base, and a kite whose four angle measures are
  $\beta/n$, $180^\circ+2(i-1)\beta/n$, $\beta/n$, and
  $180-2i\beta/n$.  The $(i-1)\beta/n$ triangles are hinged into a
  chain $D_i^a = \chain{(i-1)\beta/n}{\ell_1,\ldots,\ell_{2k}}$, and
  the kites are hinged into a \term{kite-chain} $E_i^a =
  \chain{(i-1)\beta/n,i\beta/n}{\ell_1,\ldots,\ell_{2k}}$ as in Figure
  \ref{fig:movepseudodisconnected}; for $i=1$, the $\beta/n$ triangles are hinged into the
  kite-chain $E_1^a=\chain{0,i\beta/n}{\ell_1,\ldots,\ell_{2k}}$). We
  then hinge $(D_i^a)_{0}$ and $(E_i^a)_{0}$ to point $h_i$ of
  $L_i^{a(*)}$, and hinge point $h_b$ of $L_i^b$ to $E_i^a(1)$. See
  Figure \ref{fig:movepseudodisconnected}.

  As before, this is a refinement of $G$ since each piece may take its
  original position. We now describe the alternate configuration: For
  $2\le i\le n$, chain $D_i^a$ fills in the kite-sweep of $L_{i-1}^a$,
  while $L_n^a$'s kite-sweep remains unfilled. The kite-chains
  $E_1,\ldots,E_n$ fit together to form a refinement of a chain
  $\chain{\beta}{\ell_1,\ldots,\ell_k}$ connecting the two
  halves. This is exactly the desired form, so we're done.
\end{proof}

\subsection{Putting the Pieces Together}

Now we can finally write down the proof of the desired claim for this
section:

\begin{thm}\label{thm:common-refinement}
  For any finite collection of polygons $P_1,\ldots,P_n$ of equal
  area, there exists a common refinement $C\prec P_i$ for $1\le i\le
  n$.
\end{thm}

% It suffices to show that hinge decompositions are transitive in the
% following sense:

% \begin{thm}
%   For any two hinged figures $A$ and $B$ of equal total area, there
%   exists a common refinement $C\prec A$ and $C\prec B$.
% \end{thm}

% Indeed, it easily follows from induction and Theorem~\ref{} that any
% $n$ hinged figures $A_1,\ldots,A_n$ of equal total area have a
% common refinement, so in particular, $n$ polygons $P_1,\ldots,P_n$
% have a common refinement.  Thus, we proceed to prove Theorem
% \ref{thm:transitive}:

\medskip

\ifabstract
  \begin{wrapfigure}[7]{r}{3in}
    \centering
    \vspace*{-6ex}
    \subfloat[][]{\includegraphics[scale=.65]{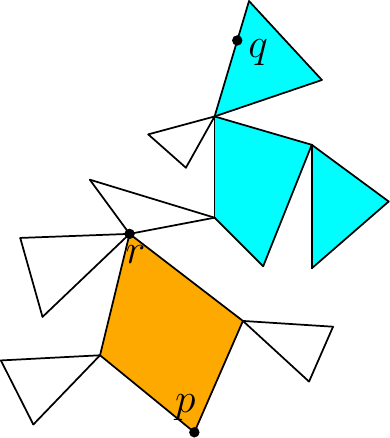}}
    \subfloat[][]{\includegraphics[scale=.65]{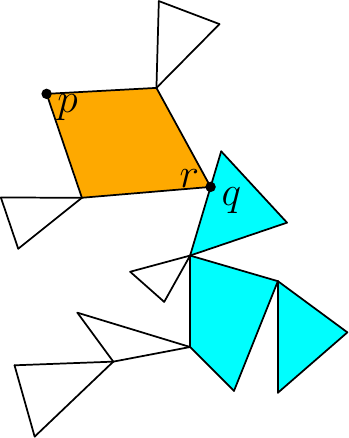}}
    \subfloat[][]{\includegraphics[scale=.65]{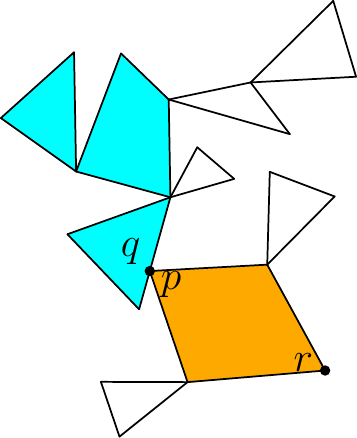}}
    \caption{Rearranging a pseudo-figure by means of rooted subtree
      movements.}
    \label{fig:pseudotwostep}
  \end{wrapfigure}

\proof[sketch]
  By the Lowry-Wallace-Bolyai-Gerwien Theorem,
  there exists a common
  decomposition of $P_1,\ldots,P_n$ into finitely many polygons
  $\{L_i\mid 1\le i\le k\}$. Hinging these together, we may 
  inductively apply the subtree movement constructions defined in
  this section as indicated in Figure~\ref{fig:pseudotwostep} to obtain a full hinged dissection. For full
  details, see the Appendix.
\endproof
\else
\begin{proof}
  By the Lowry-Wallace-Bolyai-Gerwien Theorem stated in the introduction,
  there exists a common
  decomposition of $P_1,\ldots,P_n$ into finitely many polygons
  $\{L_i\mid 1\le i\le k\}$; hinge these links together to form a
  tree-like hinged figure $A$. Now suppose we have a tree-like
  refinement $B_{t-1}$ that simultaneously refines $A$ and
  $P_1,\ldots,P_{t-1}$; we'll find a refinement $B_{t}\prec B_{t-1}$
  that is also a refinement of $P_{t}$ (the base case $t=1$ is
  realized by $A$ itself).

  Let $A_t$ be a tree-like hinging of links $\{L_i\}$ that refines
  $P_t$.  Since $B_{t-1}$ refines $A$, it suffices by repeated
  application of Theorem~\ref{thm:move_pseudosubtree} to show that
  $A_t$ may be obtained from $A$ by finitely many rooted subtree
  movements: $B_t$ is formed from $B_{t-1}$ by performing the
  corresponding refinements of $B_{t-1}$ sequentially according to
  Theorem~\ref{thm:move_pseudosubtree}.

  First, re-index the links $L_i$ so that for each $1\le m\le k$, the
  subfigure of $A_t$ formed by links $L_1,\ldots,L_m$---denoted
  $A_t|_{L_1,\ldots,L_m}$---is connected.  We rearrange $A$
  inductively.  Suppose that $A$ has been rearranged by rooted subtree
  movements into a tree-like figure $A_t^m$ so that the subfigure
  $A_t^m|_{L_1,\ldots,L_m}$ of $A_t^m$ formed by links
  $L_1,\ldots,L_m$ is equivalent to $A_t|_{L_1,\ldots,L_m}$. (We start
  with the base case $A_t^1 = A$.)  We now move $L_{m+1}$ into place.

  Let $p$ and $q$ be the boundary points of $L_{m+1}$ and
  $A_t|_{L_1,\ldots,L_m}$, respectively, which are identified in
  $A_t$. Also let $r$ be the hinge of $L_{m+1}$ closest to
  $A_t^m|_{L_1,\ldots,L_m}$ in $A_t^m$; i.e., in the incidence graph,
  $r$ is the hinge whose removal would separate $L_{m+1}$ from all
  links $L_1,\ldots,L_k$ (this vertex exists since
  $A_t|_{L_1,\ldots,L_m}$ is connected).

  \begin{figure}
    \centering
    \subfloat[][]{\includegraphics[scale=.65]{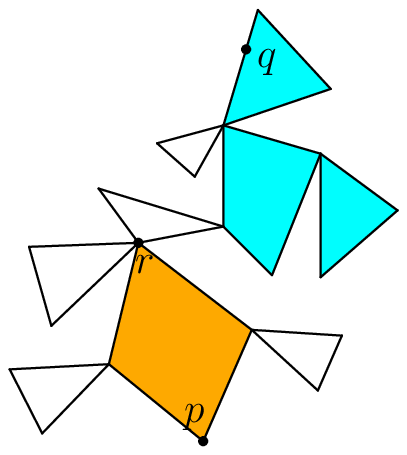}}
    \subfloat[][]{\includegraphics[scale=.65]{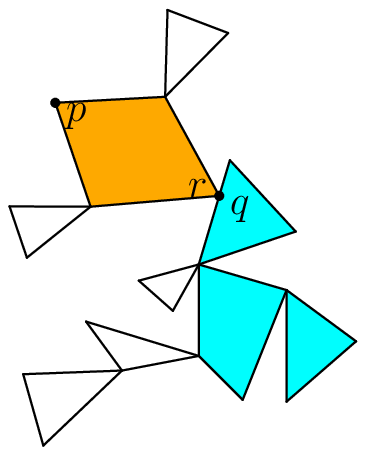}}
    \subfloat[][]{\includegraphics[scale=.65]{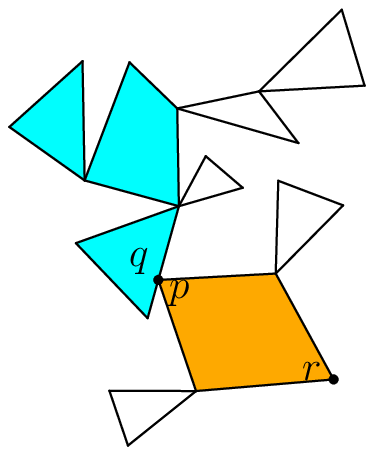}}
    \caption{Rearranging a pseudo-figure by means of rooted subtree
      movements.}
    \label{fig:pseudotwostep}
  \end{figure}

  We now perform two rooted subtree movements, as illustrated in
  Figure \ref{fig:pseudotwostep}. First, break $A_t^m$ into two rooted
  trees $(S,r)$ and $(T,r)$ at $r$, so that $L_{m+1}\in S$ and
  $L_1,\ldots,L_m\in T$; move $S$ and rejoin in the form
  $\join{S}{r}{T}{q}$. Next, break this into the same two rooted
  subtrees $(S,r)=(S,q)$ and $(T,q)$, and move $T$ to rejoin in the
  form $\join{T}{q}{S}{p}$. In this way, the links $L_1,\ldots,L_m$
  are not disturbed, and $L_{m+1}$ is rooted properly, i.e. this new
  hinged pseudofigure is $A_t^{m+1}$. Repeat this procedure to
  finally obtain $A_t^k = A_t$ by rooted subtree moves, as desired.
\end{proof}
\fi

\section{Continuous Motion}
\label{sec:motion}

Theorem~\ref{thm:common-refinement} constructs a hinged dissection
that has a configuration in the form of each of the $n$ polygons.
This section shows how to further refine that hinged dissection
to enable it to fold continuously into each polygon while avoiding
intersection among the pieces:

\begin{thm}\label{thm:motion}
  Any hinged figure $A$ has a refinement $B\prec A$ so that any two
  configurations of $B$ are reachable by a continuous
  non-self-intersecting motion.
\end{thm}

Indeed, given polygons $P_1,\ldots,P_n$ of equal area,
Theorem~\ref{thm:common-refinement} guarantees that there exists a hinged
figure $F$ that refines each of $P_1,\ldots,P_n$.
By Theorem~\ref{thm:motion}, there is a refinement $F'\prec F$ that is
universally reconfigurable without self-intersection. In particular,
$F'$ can continuously deform between any of the configurations induced
by the $P_i$s. This figure $F'$ solves the problem,
proving the first sentence of Theorem~\ref{thm:main2D}.

To prove Theorem~\ref{thm:motion}, we require two preliminary results;
the first dealing with polygonal chains and slender adornments, and
the other involving chainifying a given hinged figure.

\subsection{Slender Adornments}

Slender adornments are defined by Demaine, et al.\ in
\cite{planarshapes}.  An \term{adornment} is a connected, compact region
together with a line segment $ab$ (the \term{base}) lying inside the
region. Furthermore, the two boundary arcs from $a$ to $b$ must be
piecewise differentiable, with one-sided derivatives existing
everywhere. An adornment is a \term{slender adornment} if for every
point $p$ on the boundary other than $a$ and $b$, the primary inward
normal(s) at $p$, namely the rays from $p$ perpendicular to the
one-sided derivatives at $p$, intersect the base segment $ab$
(possibly at the endpoints).
In \cite{planarshapes}, it is shown that chains of slender adornments
cannot lock. Specifically, they show the following:

\begin{thm} {\rm \cite[Theorem~8]{planarshapes}} \label{thm:slender-adornments}
  A strictly simple polygonal chain adorned with slender adornments
  can always be straightened or convexified.
\end{thm}

(In a \term{strictly simple} polygonal chain, edges intersect each other
only at common endpoints.)  This implies
that any strictly simple polygonal open chain is universally
reconfigurable, because to find a continuous motion between two
configurations $c_1$ and $c_2$, one may simply follow a motion from
$c_1$ to the straightened configuration $c$, and then reverse a motion
from $c_2$ to $c$.

\subsection{Chainification}

Next, we prove that any hinge figure has a refinement that is
chain-like and simply adorned:

\begin{thm}\label{thm:chainify}
  Any hinge figure $F$ has a chain-like refinement $G\prec F$ so that
  $G$ consists of a chain of equally-oriented obtuse triangles hinged
  at their acute-angled vertices.
\end{thm}

\begin{proof}
  % The following two refinements are helpful.
  %   %
  % The first is \term{link splicing}, which replaces link $L =
  % L_1L_2\ldots L_n$ with two links $L_1L_2\ldots L_i$ and
  % $L_iL_{i+1}\ldots L_1$ hinged at $L_1$. The only hinge possibly
  % affected by this change is the hinge at $L_i$, which may attach to
  % either link in the correct place
  %   %
  % \term{hinge detaching} -- disconnecting a link from an incident
  % hinge, i.e. removing the corresponding edge in the incidence
  % graph, so long as the process does not disconnect the figure. See
  % Figure~\ref{}.
  First we refine $F$ to consist of a tree of triangles hinged at
  vertices, as follows.
  For each $n$-sided link $L$ with $n\ge 4$, draw a collection of
  triangulating diagonals.
  Sequentially, for each such diagonal $V_1V_i$ currently in link $V =
  V_1V_2\ldots V_k$ (which may be a refinement of an original link),
  replace $V$ with two links $V_1V_2\ldots V_i$ and $V_iV_{i+1}\ldots
  V_1$ hinged at $V_1$, attaching the hinge originally at $V_i$ to its
  corresponding position on either refined piece.  The resulting
  figure indeed consists of triangles hinged at vertices.

  Next, if the resulting triangulated figure is not tree-like, we may
  repeatedly remove an edge from a cycle in the incidence graph
  (i.e. remove the corresponding link from its hinge) until the graph
  becomes tree-like. Call this refinement $H$.

  \begin{figure}
    \centering
    \begin{minipage}{0.45\textwidth}
      \centering
      \includegraphics[scale=.65]{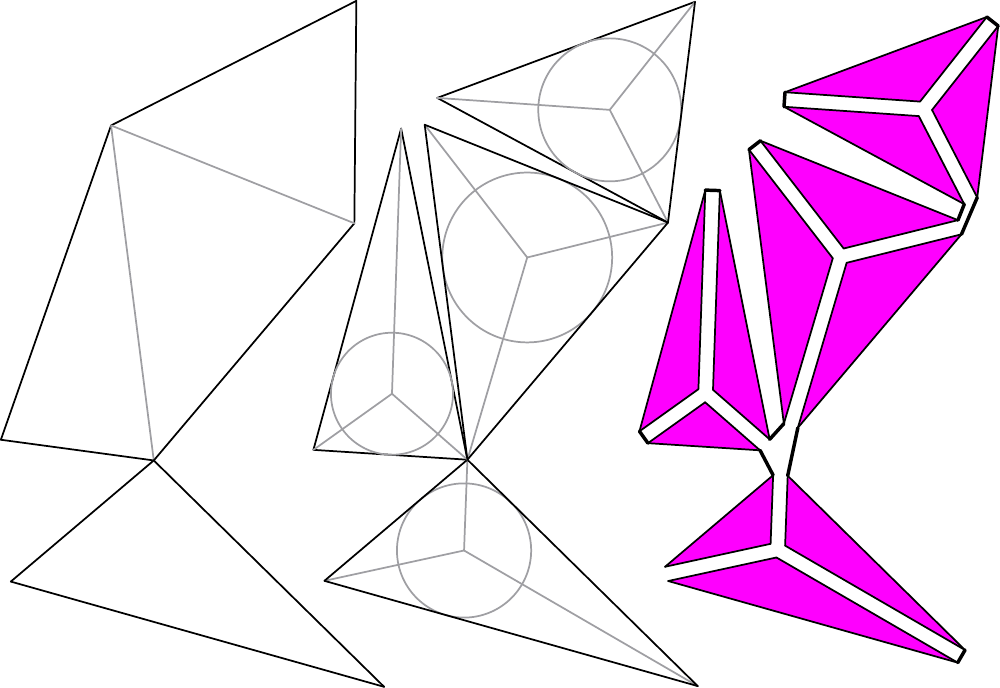}
      \caption{Chainifying a hinged figure}
      \label{fig:chainify}
    \end{minipage}
    \begin{minipage}{0.45\textwidth}
      \centering
      \includegraphics[scale=.65]{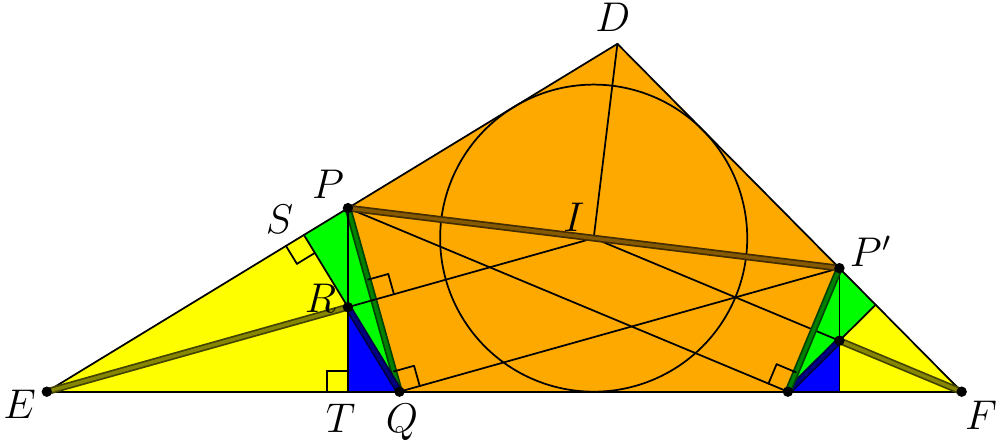}
      \caption{Refinement to hide bars from each other.}
      \label{fig:7-piece}
    \end{minipage}
  \end{figure}

  For each triangular link $ABC$ in $H$, divide $ABC$ into three
  triangles $AIB$, $BIC$, $CIA$, where $I$ is the incenter of
  $\triangle ABC$. Note that
  $
    \angle BIC = \pi - \frac12\angle B - \frac12\angle C
    > \pi - \frac12 (\angle A+\angle B+\angle C) = \frac\pi2,
  $
  i.e. $\angle BIC$ is obtuse, and likewise for the others. Finally,
  by hinging these obtuse triangles at the base vertices by walking
  around $H$'s boundary (Figure~\ref{fig:chainify}), we obtain the
  desired chain-like refinement $G$.
\end{proof}

\subsection{The Final Piece of the Puzzle}

We now prove Theorem~\ref{thm:motion}, i.e., that any hinged figure $A$
has a universally reconfigurable refinement~$B$.

\medskip

\ifabstract
\begin{proof}[of Theorem~\ref{thm:motion}]
\else
\begin{proof}[Proof of Theorem~\ref{thm:motion}]
\fi
  As shown in Theorem~\ref{thm:chainify}, $A$ has a refinement $C$
  consisting of obtuse triangles hinged along their bases. For each
  such obtuse triangle $\triangle DEF$, we create the following
  $7$-piece refinement (Figure~\ref{fig:7-piece}).
  \iffull

  \fi
  Let $I$ be the incenter of triangle $DEF$, and suppose the line
  through $I$ perpendicular to $DI$ intersects sides $DE$ and $DF$ at
  $P$ and $P'$ respectively; by obtuseness of $DIE$, $P$ lies on the
  interior of side $DE$, and likewise for $P'$. Reflect $P$ over angle
  bisector $EI$ to $Q$; it is not hard to check that $\angle PQP'$ =
  $90$. Define $R$, $S$, and $T$ as illustrated; since $PEQ$ is
  isosceles and acute, $R$ is inside $PEQ$. Repeat on the other side
  to form the $7$-piece refinement as illustrated. As the angles in
  all of the adornments are $90^\circ$ or larger, each can be easily
  checked to be slender. Furthermore, no bar can touch any other
  except at the vertices, since the bars in $DEF$ only touch the
  boundary of $DEF$ at single vertices, and no two bars within $DEF$
  are touching.  Thus, the resulting hinged figure $B$ is a strictly
  simple polygonal chain with slender adornments that refines $C$ (and
  hence refines $A$), so we are done.
\end{proof}

\section{Pseudopolynomial}
\label{sec:pseudopoly}

We now describe how to combine the preceding steps with ideas
of Eppstein \cite{Eppstein-2001} and the classical rectangle-to-rectangle
dissection of Montucla \cite{Ozanam-1778}
to perform our hinged dissection using only a
pseudopolynomial number of pieces, proving the second sentence of
Theorem~\ref{thm:main2D}. In contrast to
Theorem~\ref{thm:common-refinement}, we will only describe the transformation
between two given polygons rather than arbitrarily many.
A simple induction shows that the construction remains pseudopolynomial
for a constant number of target polygons.

The idea is as follows: the inefficiency in the preceding construction
is because movements may traverse the same hinges many times,
leading to a recursive application of pseudosubtree movement
and giving exponentially many interconnections.
By performing some simplifying steps prior to subtree movement,
we can instead ensure that movements are along mostly-disjoint paths
so that all recursion is constant-depth.

\begin{figure}
  \centering
  \subfloat[][]{\includegraphics[scale=.65]{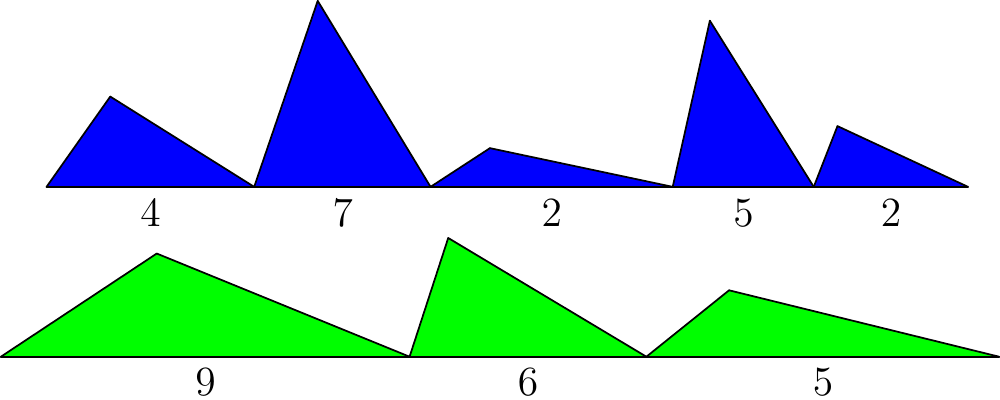}}
  \subfloat[][]{\includegraphics[scale=.65]{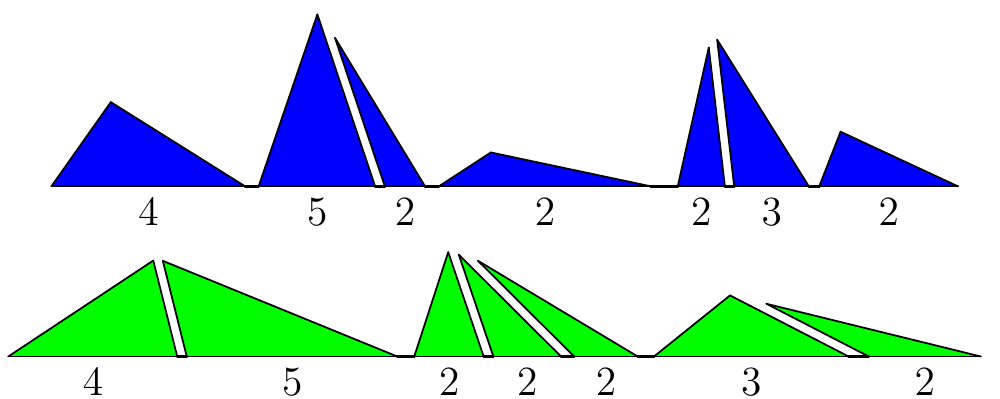}}
  \caption{Equalizing the areas of two triangular chains.}
  \label{fig:equalarea}
\end{figure}

To do this, given two figures, chainify them so we have two chains
of triangles. We then further subdivide them so that both chains
have the same number of links, and such that corresponding
triangles have the
same area. We do this using an idea from \cite{Eppstein-2001}:
cut the triangles from base to
apex along the lines that yield the desired area, hinging
at the base to maintain connectivity (see
Figure~\ref{fig:equalarea}).

Given these compatible chains, our task reduces to producing
hinged dissections between each pair of equal-area triangles
in such a way that the base vertices of one map to the base
vertices of the other. If each individual pair of triangles
requires only pseudopolynomially many pieces, we will be done.

\subsection{Pseudocuts}

Let $G \prec F$ be hinged figures. If we make
a cut in $F$, producing $F'$, we may not be able to
directly make the same cut in $G$: attempting to do so may
disconnect the figure. We give here a construction allowing
us to produce an $H$ refining both $G$ and $F'$. In keeping with
earlier terminology, we call the cut in $F$ a \term{pseudocut}
with respect to $G$.
This operation will be useful in keeping everything
pseudopolynomial.

\begin{thm}\label{thm:pseudocut}
Let $f_1$ and $f_2$ be boundary
points along some link of a tree-like figure $F$.
Let $F'$ be the tree-like figure
obtained by adding a straight-line cut between
$f_1$ and $f_2$ and hinging at $f_1$, and suppose $G\prec F$.
Then there exists a common refinement $H\prec G$ and
$H\prec F'$. Further, $H$ differs from $G$ only within
the free region of the boundaries defined by adding
the straight-line cut of $F'$ to $G$.
\end{thm}

\begin{proof}
  Consider the behavior of $G$ along the edge from $f_1$ to
  $f_2$. In $G$ the pseudocut may traverse several hinged
  pieces. Suppose first that the pseudocut hits no
  existing hinges.
  Let $\{h_i\mid 1\le i\le n\}$ be the points of
  intersection between the pseudocut and the existing
  edges of $G$ (so that in particular $h_1 = f_1$ and
  $h_n = f_2$). After the cut has been made, distinguish
  identified vertices on each side as $h_i^a$ and
  $h_i^b$. We proceed inductively along the
  segments in $H$, beginning with $h_1$ to $h_2$ which
  we can easily cut and hinge exactly as in~$F$.
  
  \begin{figure}
    \centering
    \subfloat[][]{\includegraphics[scale=.65]{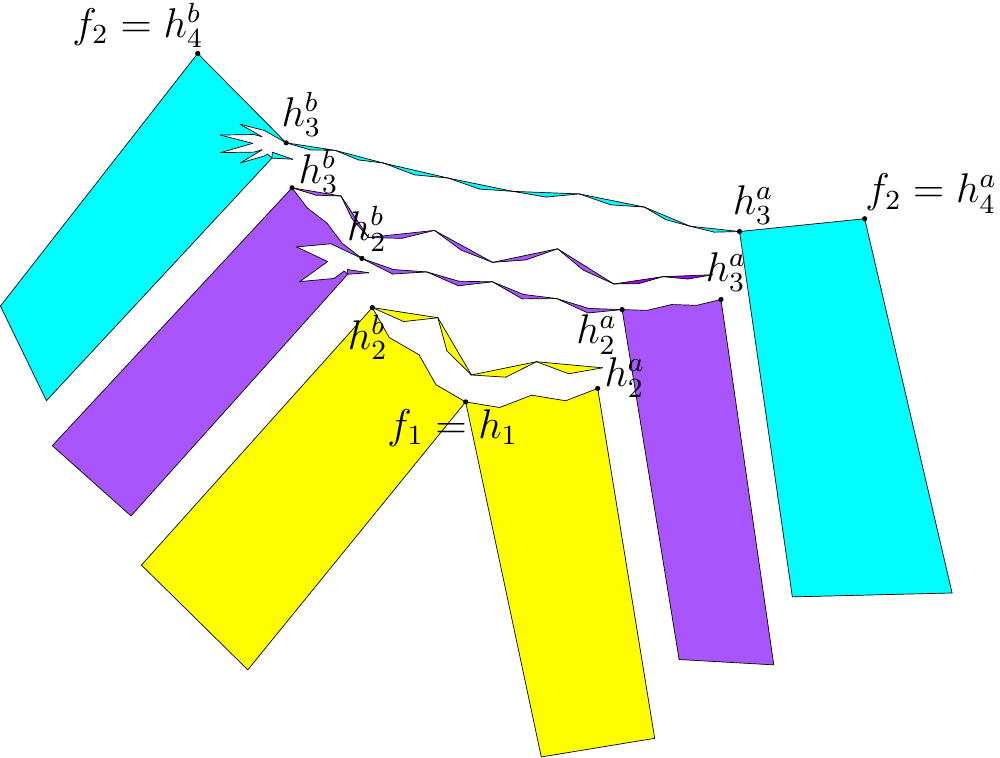}}
    \subfloat[][]{\includegraphics[scale=.65]{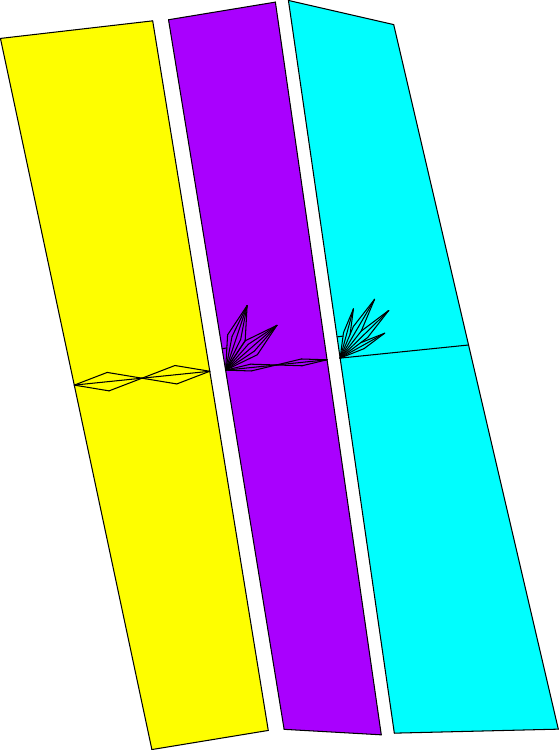}}
    \subfloat[][]{\includegraphics[scale=.65]{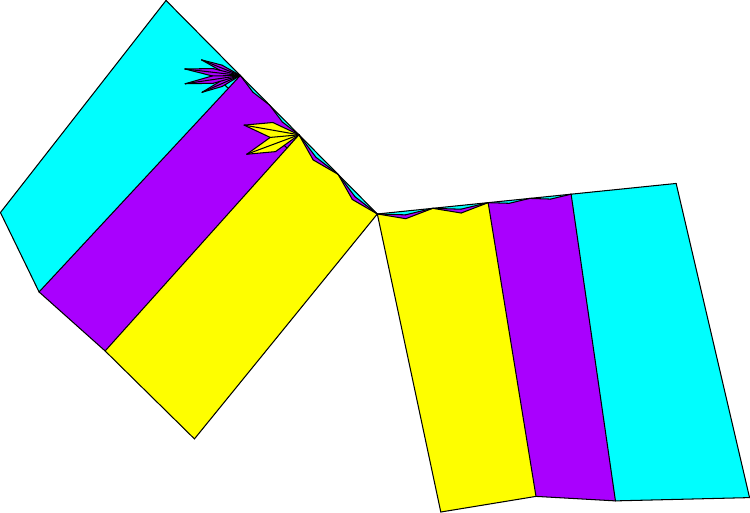}}
    \caption{Making a pseudocut across existing edges}
    \label{fig:pseudocut}
  \end{figure}
  
  Now suppose we have already modified all segments
  up to $h_i$ to refine $F$ appropriately. Cut the
  segment from $h_i^a$ to $h_{i+1}$, hinging at $h_i^a$,
  and perform a rooted subtree movement from $h_i^a$
  to $h_i^b$. We modify this movement in two ways:
  first, instead of tracing the entire exterior path between
  the two points, we use only the direct path
  along the cut line, with the intermediate vertices
  as base points of our triangle chain. Second,
  since this path will cross the paths used by previous
  segments, we reuse all
  base points from earlier cut-out triangles,
  decreasing the angle slightly to separate them;
  see Figure~\ref{fig:pseudocut} for an example of this construction.
  Repeating this for all segments, then, we obtain
  the full pseudocut as desired.
  
  Now consider the case where
  one or more hinges of $G$ lie on the cut edge. We only
  need that our inductive step can cut hinges as well
  as simple links. Where before our inductive
  transformation was based on subtree movement, for
  this case we will use pseudosubtree movement.
  Since we have already covered cutting links, we may here
  consider only links entirely on one side of the cut
  line. Treat all such links as a rooted pseudosubtree
  and again perform pseudosubtree movement traversing
  only the cut path instead of the entire figure
  boundary, and again reusing previous boundary triangle
  base points. The rest of the argument is identical.
  
  Combining these two inductive
  steps allows us to produce the desired
  refinement across any existing configuration of the cut line in
  $G$, so we are done.
\end{proof}

\subsection{Rectangle to Rectangle}

\begin{figure}
  \centering
  \subfloat[][Swinging $B$ back and forth to cover $A$]{
    \includegraphics[scale=.65]{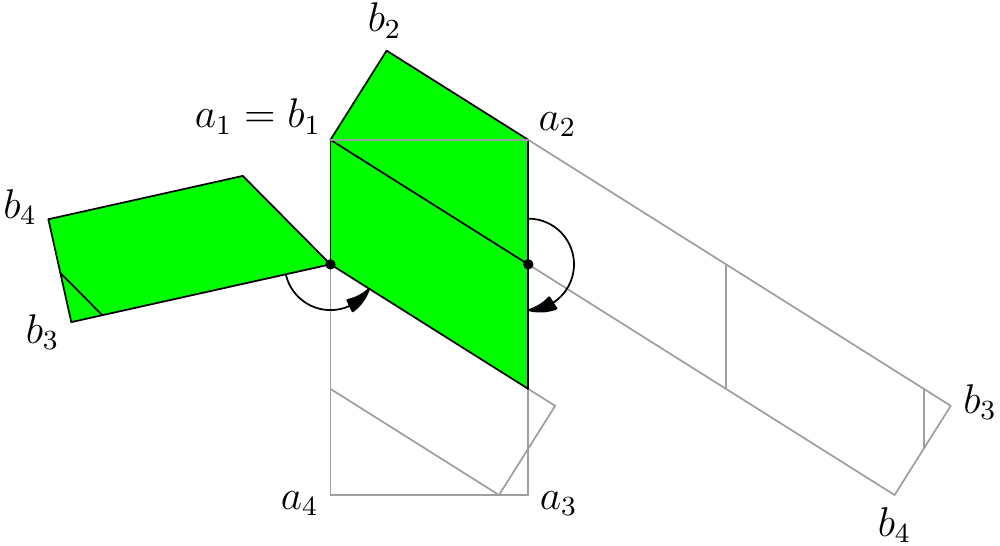}
    \label{fig:recttorect_snake}}\hfil
  \subfloat[][Capping the top and bottom of $B$]{
    \includegraphics[scale=.65]{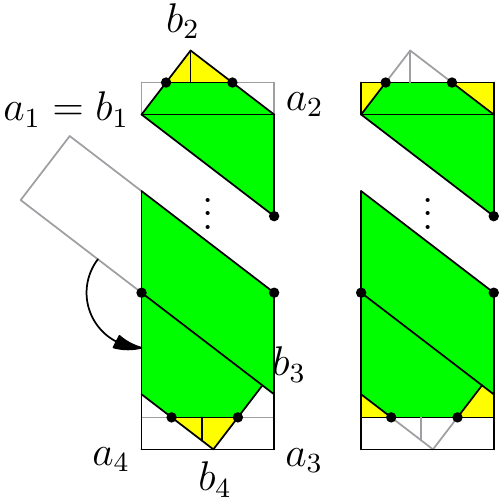}
    \label{fig:recttorect_cap}}\hfil
  \subfloat[][Moving the end of $B$ to ``wrap around'' $A$]{
    \includegraphics[scale=.65]{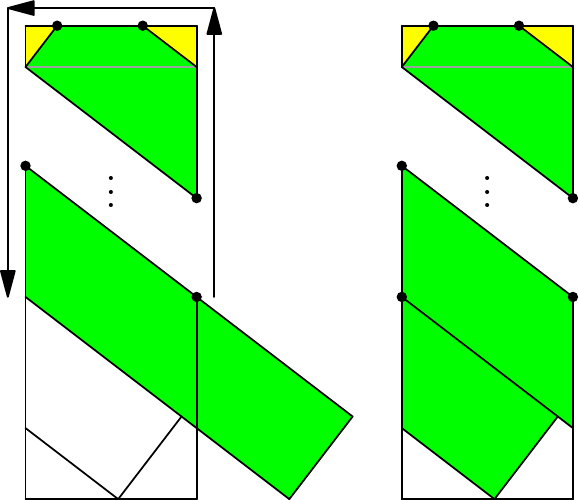}
    \label{fig:recttorect_wrap}}
  \caption{The stages of the rectangle-to-rectangle transformation.}
\end{figure}

\ifabstract

Our construction also requires an efficient
transformation between arbitrary equal-area
rectangles. Since this dissection is nearly identical
to the 230-year-old classical dissection, we outline this
only briefly.

First (Figure~\ref{fig:recttorect_snake}), snake the narrower
rectangle ($B$) back and forth, covering as much of the other
rectangle ($A$) as possible. Second
(Figure~\ref{fig:recttorect_cap}), cut the triangular portions of
$B$ and fold them into rectangular caps, if possible.
Third (Figure~\ref{fig:recttorect_wrap}), if this was
not possible then we need to move the extended portion
of $B$ to the left side of $A$ as if it ``wrapped around''
when hitting the right edge. After this, the second step
is guaranteed to be successful. For full details of this
construction, see the Appendix.

\else

The next step is to describe a pseudopolynomial hinged
dissection between any two equal-area rectangles. We will
compose this with our pseudocut operation in the next
section to dissect between any two triangles as well.
This step is heavily based on the classical (non-hinged)
dissection, modified using our subtree
operations to allow all operations to be hinged.

Take two rectangles of equal area, $A$ and $B$, and
suppose $B$ has the smallest minimum side length.
We begin by aligning both rectangles with their shorter
edge on the horizontal axis
and longer edge on the vertical axis, and identifying
the two top left vertices. We then rotate $B$ counterclockwise
until its lowest vertex is horizontally aligned with the base of $A$.
Label the vertices $a_i$ and $b_i$ for $1 \leq i \leq 4$, starting
at the top left and moving clockwise.

At this point the horizontal cross-section of both rectangles
is equal. Now, cut
$B$ along edge $a_2 a_3$ (the right side), hinging at the bottom of
the cut, and rotate the extended portion of $B$ clockwise
by $\pi$ to cover a strip of $A$. Again, cut $B$, now along
the left side $a_1 a_4$, hinging at the bottom, and rotate it back in,
covering another horizontal strip of $A$. Continue in this
way until the remaining segment of $B$ extending past $A$'s
boundary is no longer enough to cover an entire horizontal
strip (see Figure~\ref{fig:recttorect_snake}).

Now consider the subtriangle $a_1 b_2 a_2$ of $B$. We cut
it horizontally at half its height, and vertically from
$b_2$, and rotate the resulting components out
to form a rectangular cap with the same width as $A$
(see Figure~\ref{fig:recttorect_cap}).

There are now two cases: either the remaining segment of $B$
extends up and left, or down and right. In the fortunate
former case, we can perform the entire transformation using
only classical-style manipulations: swing the extended portion
back into $A$. This will give a ``triangular'' base, similar
to the triangle that was on top of $A$, but offset horizontally
and wrapping through the edge of $A$. We can make this
rectangular as well by a nearly identical transformation:
cut horizontally at its vertical midpoint, and vertically
as shown in Figure~\ref{fig:recttorect_cap} so that the pieces line up
with the border of $A$ when we swing them out.
A quick case analysis shows that this always works.

Now consider the remaining case where the end of $B$ extends
down and to the right. We can directly reduce this to the
previous case: cut $B$ along edge $a_2 a_3$, hinging at
the {\it top}, and use rooted subtree movement to move this
subtree counterclockwise around the figure to line up with
the left side of $A$ (Figure~\ref{fig:recttorect_wrap}). The configuration
of the base is now as
though $B$ had extended up and left and we rotated the
extended piece down
into $A$ as before, the only change being the kite sweep
at the top vertex. Since this kite sweep can easily
be kept above the vertical midpoint of the triangle,
it doesn't interfere with our new cuts, and 
the same capping strategy works
without alteration.

With these steps, $B$ is transformed into a rectangle having the
same area and width (and therefore same height) as $A$.
\fi

\subsection{Unaltered Subtree Movement}

It will be useful in the analysis to be able to perform
subtree movements without modifying the subtree. This is
in contrast to earlier constructions, which cut the
kite tree out of the subtree being moved.

\begin{figure}
  \centering
  \subfloat[][]{\includegraphics[scale=.65]{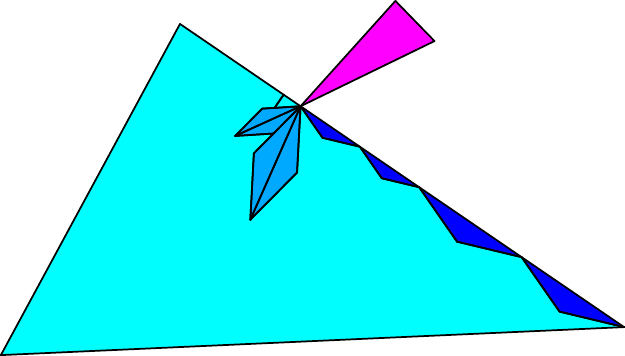}}\hfil
  \subfloat[][]{\includegraphics[scale=.65]{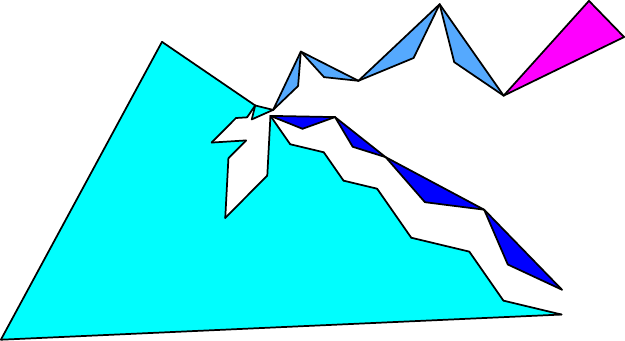}}\hfil
  \subfloat[][]{\includegraphics[scale=.65]{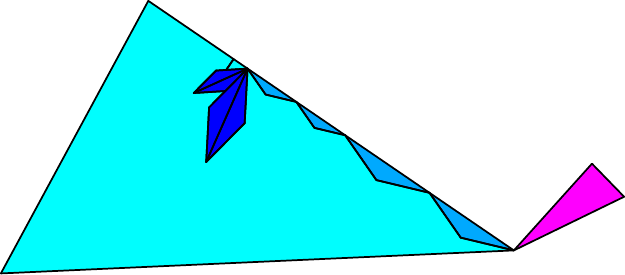}}
  \caption{Moving a subtree by modifying only the parent tree}
  \label{fig:trickymove}
\end{figure}

Accomplishing this is a simple modification of the
earlier operation: we cut both the kite sweep and the
triangle chain out of the free region of the parent
subtree, and only need to alter the hinge connection
points. The kite sweep will form a chain connecting
at one end to the source vertex and at the other
end to the subtree being moved, and the triangle
chain will be a loose chain hanging from the
source vertex (geometrically; it is not connected
directly to the kite sweep); see Figure~\ref{fig:trickymove}.
To move the subtree, we then extend the kite sweep
out, fill it in with the hanging triangle chain,
and place the moving tree at the destination point.

\subsection{Polygon to Polygon}

With the pieces described, the transformation is
simple: first, perform the equal-area chainification
on both input polygons. Then convert each triangle
to a rectangle using the same cutting procedure
described in
Figure~\ref{fig:recttorect_cap}
for capping the top of the rectangle.
We would then like to map between the two rectangles using the
rectangle-to-rectangle transformation. However, to ease
analysis, we actually view the rectangle-to-rectangle
transformation as being done first, and then
transform the rectangles back into the original
triangles by making the necessary pseudocuts as though
the figures were a solid rectangle.

After these steps, we will have pairwise dissections
between the triangles in the chain. In the last step,
we use
Unaltered Subtree Movement to adjust the hinge joints
between pairs of triangles to lie on the correct
boundary points. This yields a common refinement of the
two triangle chains, and we are done.

\subsection{Analysis}

\begin{thm}\label{thm:pseudopoly}
The procedure described above gives a dissection with a
pseudopolynomial number of pieces.
\end{thm}

\ifabstract
The proof of this theorem is deferred to the Appendix.
\else
\medskip
\begin{proof}
Since our construction worked independently on each pair
of equal-area triangles, and the equal-area chainification
step itself was pseudopolynomial, it suffices to show that
our transformation is pseudopolynomial when applied to
a single pair of triangles.

We analyze the maximum number of pieces produced by
our dissection by considering the related question of the
smallest possible non-zero value that can be computed
at any intermediate step of the dissection. This value
gives a lower bound on, for instance, the smallest distance
between any two distinct vertices. If this distance is at
least inverse polynomial, then a simple area argument shows
that the number of pieces in the dissection is also at
most polynomial.

Our task, then, is to show that the smallest non-zero
value produced during the computation is indeed inverse
polynomial, a question that can be dealt with almost
entirely algebraically. We will use the bounds from
\cite{mehlhorn}. Under these bounds,
it suffices to show that if we view all numerical computations
as a DAG, then algebraic extensions, multiplications,
and divisions (by previously computed expressions) are never
nested to more than constant depth, and that addition
and subtraction, and multiplication or division by fixed
constants, are never nested to more than linear depth.
Under these constraints, if we are still able to compute
the coordinates of all vertices on our figure, we will be done.
For this section, refer to addition and subtraction, and
multiplication and division by fixed constants, as
simple arithmetic operations.

It is important in maintaining constant depth of
multiplications that we never rotate the same point
or vector more than a constant number of (nested)
times. For instance, while for simplicity we have
described adjacent kites as touching along a common
angle, to attain a pseudopolynomial bound we cannot
do this since each consecutive kite would give another
nested rotation. 
%Thus, for example, in our kite sweeps, adjacent
%kites will not meet along a common edge, since computing
%the sequence of coordinates induced by these adjacent
%angles would involve many sequential rotations.
Instead,
whenever we need to perform a rotation, we will use
a rational approximation of the angle. We can approximate
an angle within an arbitrary error bound $\delta$
by using the ratios of integers that are polynomial in
$1/\delta$ (for instance by solving for an appropriate
Pythagorean triple). We can further ensure that we always
have plenty of room for such approximations by restricting
ourselves to free regions that use only half of the true
available angular space, so we know that there will also
be polynomial error tolerance built into any desired
angles, so any single rotation will still only involve
polynomial values. Kite sweeps will thus be spread out
near each other but with angles ensuring they do not
intersect.

Now let us consider the arithmetic/algebraic depth of the
expressions produced during each step of our transformation.
As mentioned above, for the purposes of analysis we
consider the procedure as beginning with two equal-area
rectangles and then map back to the original triangles
by adding the appropriate pseudocuts.

First: When we rotate the narrower rectangle to have the
same horizontal cross-section as the first, this requires a
quadratic extension, as well as a constant number of
multiplications and divisions to rotate all the vertices of
the second rectangle.

Second: snaking the second rectangle back and forth along
the first (Figure~\ref{fig:recttorect_snake}) requires
only additions and subtractions, since
all rotations are by $\pi$. Since each pass of the rectangle
must cover a strip whose height is at least the width of
the second rectangle (which is at least half the height
of its originating triangle), we require a linear number
of passes,
so the depth of simple arithmetic operations here is at
most linear.

Third: Capping the triangles at the top and bottom
(Figure~\ref{fig:recttorect_cap}). This also
requires only rotations by $\pi$, except for the extra step
of finding the midpoints of the triangles' ascent, which is just a
simple arithmetic operation (division by 2).

Fourth: Moving $B$'s subtree back to the left side of
$A$ (Figure~\ref{fig:recttorect_wrap}), if
necessary. The base vertices of the triangles cut out of the
border of the figure can easily be placed on rational points
with low relative denominator.
We will choose the inner angle of the triangles, as described
earlier, to remain within the free area along the border while
requiring a constant number of multiplications and divisions
with suitable rational numbers. After
choosing the interior angle, finding the interior vertex of a
triangle can be done by intersecting its two edges, which
also takes a constant number of multiplications and divisions,
and can be done independently for each internal vertex to
prevent nesting. Note similarly that each edge traversed by
the path can be handled independently.

This still leaves the kite sweep inside the moved subtree.
We can produce this from the triangle chain:
each triangle pair can be made into a kite by rotating
by $\pi$. Once this is done for each triangle pair,
we place the kites inside the appropriate vertex by
choosing approximate rotation angles as described
above so the kites are non-overlapping. Each rotation requires
a constant number of multiplications per vertex, but again each
kite can be handled independently (and uses only expressions
from its originating triangles, which were themselves computed
independently of other triangle pairs), so we still preserve constant
depth in our computations.

After this movement is complete, the bottom is capped identically to
the previous step.

Fifth: The pseudocuts. We consider these independently. Since
there are only four of them, if we show that any single pseudocut
increases the simple arithmetic depth of an existing figure
at most linearly, and the depth of remaining operations by
at most a constant, we will be done.

Consider a single pseudocut. We know that it intersects with
at most a polynomial number of edges for the simple
inductive reason that so far there are only polynomially many.
Furthermore, it lies in our existing algebraic
extension, since all cuts were simple rational cuts with
respect to the original input triangles, and the transformation
to rectangles involved only rotations by $\pi$ and the single
aligning rotation we performed in the first step. Thus, the
coordinates defining the cut are so far at only constant
computational depth.

Now, the initial vertices we need to add to our figure
are the intersections
of the cut line and any incident edges in the figure.
Each of these can be computed independently of the rest,
and each requires a constant
number of multiplications and divisions. After this we need to
cut out the triangle chains along the boundary of the cut, as
well as the kite sweep at each joint. Once more we exploit the
fact that each kite can be dealt with mostly independently:
the bases of the triangles are simple rational points along
the cut edges. The interior angles can be approximated as before,
although now we may need to compute a polynomial number of them
because of the nested triangles. However, in that case each
interior angle can also be dealt with independently, so even
though there are many interior vertices corresponding to each
base line, they are all independently still at constant
computational depth. The kite sweeps are dealt with exactly
as above, a rotation by $\pi$ followed by an approximate
rotation into the appropriate vertex while avoiding intersections.
As all of these are still independent (relying only on the
constant-depth computation to produce the appropriate base
vertices,
plus the constant-depth computation to produce the matching
interior vertex, plus the constant-depth computation to produce
a suitable rotation angle for the kite), all of this is still
done in only constant depth.

Sixth: Unaltered subtree movement. This follows from
the same argument as the ordinary subtree
movement from the fourth step. The changes in the specific
vertex used for the kite sweep make no difference to the
computational depth required. Since there are only two
such subtree movements, the added simple arithmetic
depth is again linear, and depth of remaining operations
is again constant.

Thus, we see that the computations of all steps together
remain within the required bounds, and thus all vertices
are indeed at least an inverse polynomial distance apart.
\end{proof}
\fi

\section{Three Dimensions}
\label{sec:3D}

We now consider hinged figures in three dimensions.  A \term{3D hinged
  figure} is a collection of simple polyhedra called \term{links}
hinged along common positive-length edges called \term{hinges}. As
before, the cyclic order of links around a hinge must remain constant.

Not every two polyhedra of equal volume have a common dissection.
Dehn \cite{Dehn-1900} proved an invariant that must necessarily match
between the two polyhedra.  For example, Dehn's invariant forbids
any two distinct Platonic solids from having a common dissection.
Many years later, Sydler \cite{Sydler-1965} proved that polyhedra $A$ and $B$
have a common dissection if and only if $A$ and $B$ have the same volume and
the same Dehn invariant.  Jessen \cite{Jessen-1968} simplified this proof
by an algebraic technique and generalized the result to 4D polyhedral solids.
(The 5D and higher cases remain open.)  Dupont and Sah \cite{Dupont-Sah-1990}
gave another proof which illustrates further connections to algebraic
structures.

Clearly, if two polyhedra have no common dissection, then they also have
no common hinged dissection.  We show the converse: given a common dissection
of polyhedra $A$ and~$B$, we can construct a common hinged dissection of $A$
and~$B$.  More generally, we have the following 3D analog of
Theorem~\ref{thm:common-refinement}:

\begin{thm} \label{thm:3D-common-refinement}
  Given $n$ polyhedra $P_1,\ldots,P_n$ of equal volume and equal Dehn
  invariant, there exists a hinged figure $H$ such that $H\prec P_i$
  for $1\le i\le n$.
\end{thm}

Note that our algorithms assume that the (unhinged) dissection is given.
None of the proofs that Dehn's invariant is sufficient are explicitly
algorithmic, so it remains open whether one can compute a dissection
when it exists.  (We suspect, however, that this may be possible
by suitable adaptation of an existing proof.)

All of the following definitions are 3D analogs of the definitions
given in Section~\ref{sec:terminology}.  The \term{boundary} $\partial
A$ of a hinged figure $A$ is the $2$-manifold (or collection of
disjoint $2$-manifolds) formed by identifying faces of links as
follows: (1) for each non-hinge edge $e$ of a link $\ell$, the two
faces adjacent to $e$ are connected along their common edge, and (2)
for each hinge edge $e$, each pair of adjacent faces of adjacent links
around $e$ are joined along their common edge.  The \term{incidence
graph} of a hinged figure, the notions of \term{tree-like} and
\term{chain-like}, and the concept of \term{refinement} are unchanged.

The proof will be as follows: First we will describe a revised
notion of \term{free-regions} for tetrahedra. Next, we illustrate the
technique for moving rooted subtrees and for moving rooted
pseudosubtrees, under the assumption that each link is a
tetrahedron. By tetrahedralizing the links before each pseudosubtree
movement, these assumptions lose nothing. The rest of the proof
remains unchanged.

\subsection{Defining Free Regions}

We begin by defining free regions for a tetrahedron $T$. Choose an
angle $\alpha$ smaller than the smallest dihedral angle of $T$'s six
edges. For each face $\phi$ of $T$, let $\freee_T(\phi)$ be the
tetrahedron inside $T$ whose base is $\phi$ and whose base dihedral
angles are $\alpha/3$.

For each edge $e$ of $T$, construct a cylinder $C_e$ of length
$\frac{2}{3}|e|$ centered at the midpoint of $e$ with axis along $e$.
Each cylinder has radius $r$, chosen small enough so that these six
cylinders do not intersect, and also so that for each edge $e$, $C_e$
does not intersect $\freee_T(\phi_3(e))$ and $\freee_T(\phi_4(e))$ where
$\phi_3(e)$ and $\phi_4(e)$ are the faces not adjacent to $e$.

Let $\freee_T(e)$ be the wedge of $C_e$ of angle $\alpha/3$ centered
within the dihedral angle of $T$ at $e$. By choice of $\alpha$ and
$r$, $\freee_T(e)$ will not intersect $\freee_T(\phi)$ for any face
$\phi$. These ten regions $\freee_T(\cdot)$ are the desired free
regions for tetrahedron $T$.

\subsection{Moving Rooted Subtrees}

There are two ways to join a pair of rooted subtrees $(A,a)$ and
$(B,b)$ (where $a$ and $b$ are hinges or edges of their respective
figures), as each edge has two possible orientations.  For each rooted
subtree movement of $(A,a)$ from $(B,b)$ to $(B,b')$, we will treat
$a$, $b$, and $b'$ as \emph{oriented} edges, and we will join them so
that the orientations of the joined edges match.

We may now illustrate the analog of Theorem~\ref{thm:move_subtree}:
\begin{thm}\label{thm:move_3d_subtree}
  For any two tree-like hinged figures $F$ and $F'$ related by the
  rooted subtree movement of $(A,a)$ from $(B,b)$ to $(B,b')$ for
  oriented edges $a$, $b$, and $b'$, there is a common refinement $G$,
  i.e. $G\prec F$ and $G\prec F'$.
\end{thm}

\begin{proof}
The proof is in three parts, which we outline below.

%\paragraph{Choosing the Boundary Path}
We first choose the boundary path.
Let $p_a$ be the point $1/3$ across edge $a$, and let $\phi_0$ be the
face of $F$ to the left of edge $a$, i.e. the face so that edge $a$
traces its boundary counterclockwise (as seen from the
outside). Likewise, let $p_{b'}$ be the point $1/3$ across $b'$, and
$\phi_1$ the face to $b'$'s right. Triangulate the boundary $\partial
F$, and let $\gamma$ be a piecewise linear path along $\partial F$
from $p_a$ to $p_{b'}$ that passes through no vertices of the
triangulation, crosses each edge of each triangle orthogonally (as
$\partial F$ is locally flat around each edge), begins in face $\phi_0$, and ends in face $\phi_1$. Without loss of
generality, $\gamma$ does not cross itself, as loops may be
eliminated. We may also assume that all turn angles of $\gamma$ are at
most $90^\circ$, by truncating sharp turns.

%\paragraph{Thickening the Boundary Path}
Next, we thicken the boundary path.
We now choose a small $\omega$ and form two paths $\gamma_\ell$ and
$\gamma_r$ by offsetting $\gamma$ by a constant width of $\omega$ to
the left and right, respectively.  The value $\omega$ is chosen small
enough to satisfy the following conditions: (1) $2\omega$ is smaller
than the free-radius of $a$; (2) $2\omega < \min\{|a|,|b|\}/3$; (3)
paths $\gamma_\ell$ and $\gamma_r$ have the same number of segments as
$\gamma$; (4) the region $\Gamma$ between $\gamma_\ell$ and $\gamma_r$ contains
no vertices of the triangulation of $\partial F$ and does not
intersect itself.  In essence, we have just thickened the path
$\gamma$ to have width $2\omega$.

\ifabstract
\begin{wrapfigure}{r}{2.3in}
\else
\begin{figure}%{r}{2.3in}
\fi
  \centering
  \vspace{-1.5ex}
  \includegraphics[scale=0.65]{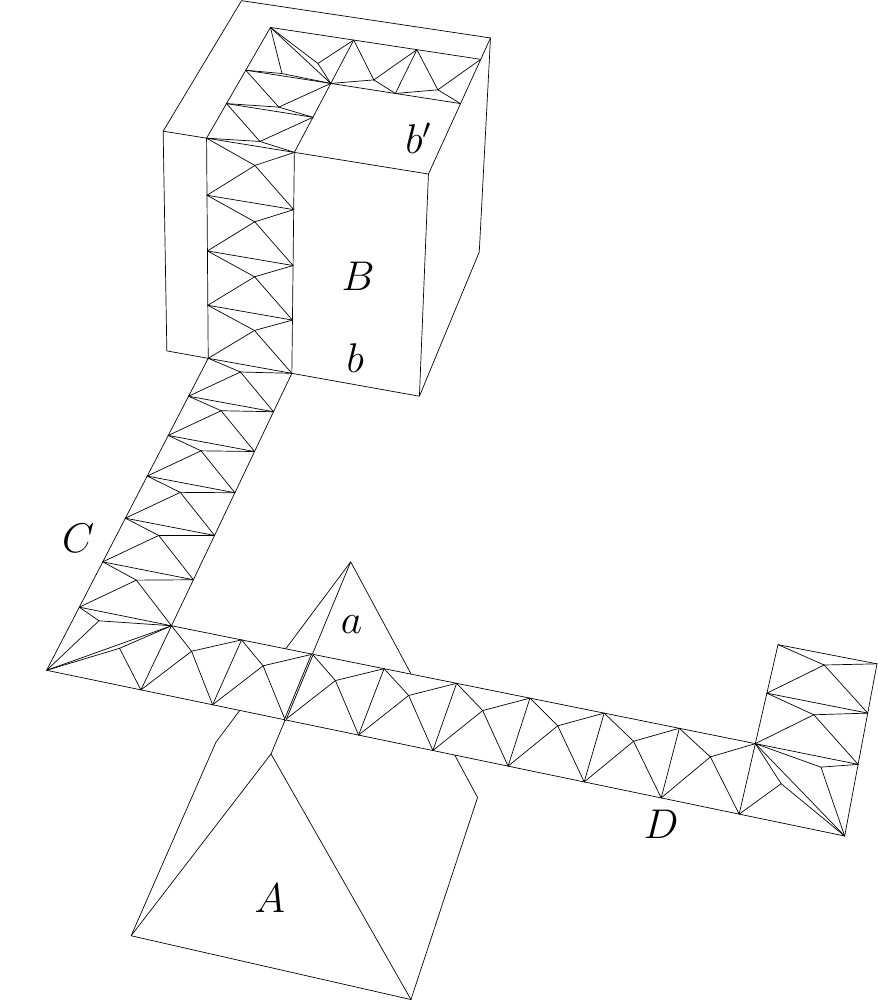}
  \caption{Moving rooted subtree $(A,a)$ from edge $b$ to edge $b'$.}
%  \vspace{-2cm}
  \label{fig:threed}
\ifabstract
\end{wrapfigure}
\else
\end{figure}
\fi

%\paragraph{Forming Chains}
Finally, we build a pyramid chain whose base is the region $\Gamma$,
as follows. We first divide $\Gamma$ into regions that will serve as
the bases. Each time $\Gamma$ crosses an edge of the triangulation of
$\partial F$, draw the intersection of $\Gamma$ with this edge; it
intersects $\gamma$ at right angles. At each rightward turn of
$\Gamma$, let $v_r$ be the vertex of $\gamma_r$ at this turn and
$v_\ell$ that of $\gamma_\ell$; draw the perpendiculars from $v_r$ to
the two edges incident with $v_\ell$, as well as segment
$v_rv_\ell$. Perform this procedure in mirror image for all leftward
turns of $\Gamma$.
These drawn segments divide $\Gamma$ into rectangles and pairs of
congruent right triangles. We subdivide each such rectangle of
dimensions $2\omega\times s$ into $2k$ rectangles of dimensions
$2\omega\times \frac{s}{2k}$, where $k$ is chosen large enough so that
$\frac{s}{2k} \le 2\omega$.
% The regions thus formed along $\Gamma$ form the bases of the
% pyramids.

We now carve out pyramids based at each of these regions along
$\Gamma$.  Let $\beta$ be the free region angle at edge $a$ of $A$.
For some sufficiently small $h$ (to be specified soon), form for each
region $R$ along $\Gamma$ the pyramid whose base is $R$ and whose vertex is at
height $h$ above the center or centroid of region $R$. These
pairwise-congruent pyramids hinge along their common edges to form a
chain $C$ of pyramids whose base is $\Gamma$. This chain $C$ may be
folded into a ``\term{kite-sweep}'' of octahedra and triangular
bipyramids at a common segment $t$ having length $2\omega$. If $h$ is
small enough so that the total dihedral angle around $t$ is at most
$\beta$, then this chain can be seen to fit within a cylindrical wedge
of radius $2\omega$, angle $\beta$, height $2\omega$, and axis along
$t$. Thus, carve out a chain $D$ duplicate to $C$ from $A$ based at
the start of $\Gamma$ in this folded form. To ensure that $A$ is
formed by simple polyhedra, we refine this link into a hinging of
tetrahedra, which is possible by \cite{Chazelle-1984}.

Finally, hinge the mutilated $A$, the mutilated $B$, and chains $C$
and $D$ as illustrated in Figure~\ref{fig:threed}; for the same
reasons as in Theorem~\ref{thm:move_subtree}, this hinged figure forms
a refinement of both $F$ and $F'$.
\end{proof}

\subsection{Moving Rooted Pseudosubtrees}

A generalization of Theorem~\ref{thm:move_pseudosubtree} follows
quickly, along the lines of our generalization of
Theorem~\ref{thm:move_subtree} described in the previous
subsection. It is straightforward to obtain the following result:

\begin{thm}
Take tree-like figures $F$ and $F'$ related by the rooted-subtree
movement of $(A,a)$ from $(B,b)$ to $(B,b')$ as in
Theorem~\ref{thm:move_3d_subtree}, and suppose $G\prec F$. Then there
exists a common refinement $H$ of $G$ and $F'$.
\end{thm}

The rest of the proof of Theorem~\ref{thm:3D-common-refinement}
follows the proof of Theorem~\ref{thm:common-refinement} unchanged.
We obtain Theorem~\ref{thm:main3D} as an immediate corollary.

\subsection{Higher Dimensions}

We believe, although we have not verified, that our techniques generalize
further to refining dissections of polyhedral solids in arbitrary dimensions
into equivalent hinged dissections.  Again we obtain only configurations,
not folding motions, for each desired polyhedral solid.
Also, it is unknown when common (unhinged) dissections exist in 5D and higher
\cite{Dupont-Sah-1990}, although the solution in 4D is again the Dehn invariant
\cite{Jessen-1968}.

%Also, it is still unknown how many pieces are necessary for a common
%hinged dissection, even in two dimensions.  This bound must depend not
%only on the number of vertices in the polygon but also on the
%complexity of the polygons themselves: consider hinge-dissecting (or
%even simply dissecting) a $1\times 1$ square into a $\frac{1}{k}\times
%k$ rectangle, with a total of $8$ vertices.  Each link can have
%diameter at most $\sqrt{2}$, so to cover the side length $k$ we require
%at least $\frac{k}{\sqrt{2}} = \Omega(k)$ pieces, a value linear in
%the size of the vertex coordinates. The authors hope to soon prove
%that the number of pieces is at most pseudopolynomial in this sense.

\section*{Acknowledgments}

This work arose during a series of open-problem sessions for
an MIT class on Geometric Folding Algorithms (6.885 in Fall 2007).
We thank the other participants of those sessions for providing
a productive and inspiring environment.
%% xxx for final version:  We also thank Stefan Langerman for helpful discussions.

% Decrease the space between bibliography items.
\ifabstract
\let\realbibitem=\bibitem
\def\bibitem{\par \vspace{-1.2ex}\realbibitem}
\fi

\begin{small}

\bibliography{hinged}
\bibliographystyle{alpha}

\end{small}

\ifabstract

\appendix

\section{Appendix}

\subsection{Details of Main Theorem}

\begin{proof}[of Theorem~\ref{thm:common-refinement}]
  By the Lowry-Wallace-Bolyai-Gerwien Theorem stated in the introduction,
  there exists a common
  decomposition of $P_1,\ldots,P_n$ into finitely many polygons
  $\{L_i\mid 1\le i\le k\}$; hinge these links together to form a
  tree-like hinged figure $A$. Now suppose we have a tree-like
  refinement $B_{t-1}$ that simultaneously refines $A$ and
  $P_1,\ldots,P_{t-1}$; we'll find a refinement $B_{t}\prec B_{t-1}$
  that is also a refinement of $P_{t}$ (the base case $t=1$ is
  realized by $A$ itself).

  Let $A_t$ be a tree-like hinging of links $\{L_i\}$ that refines
  $P_t$.  Since $B_{t-1}$ refines $A$, it suffices by repeated
  application of Theorem~\ref{thm:move_pseudosubtree} to show that
  $A_t$ may be obtained from $A$ by finitely many rooted subtree
  movements: $B_t$ is formed from $B_{t-1}$ by performing the
  corresponding refinements of $B_{t-1}$ sequentially according to
  Theorem~\ref{thm:move_pseudosubtree}.

  First, re-index the links $L_i$ so that for each $1\le m\le k$, the
  subfigure of $A_t$ formed by links $L_1,\ldots,L_m$---denoted
  $A_t|_{L_1,\ldots,L_m}$---is connected.  We rearrange $A$
  inductively.  Suppose that $A$ has been rearranged by rooted subtree
  movements into a tree-like figure $A_t^m$ so that the subfigure
  $A_t^m|_{L_1,\ldots,L_m}$ of $A_t^m$ formed by links
  $L_1,\ldots,L_m$ is equivalent to $A_t|_{L_1,\ldots,L_m}$. (We start
  with the base case $A_t^1 = A$.)  We now move $L_{m+1}$ into place.

  Let $p$ and $q$ be the boundary points of $L_{m+1}$ and
  $A_t|_{L_1,\ldots,L_m}$, respectively, which are identified in
  $A_t$. Also let $r$ be the hinge of $L_{m+1}$ closest to
  $A_t^m|_{L_1,\ldots,L_m}$ in $A_t^m$; i.e., in the incidence graph,
  $r$ is the hinge whose removal would separate $L_{m+1}$ from all
  links $L_1,\ldots,L_k$ (this vertex exists since
  $A_t|_{L_1,\ldots,L_m}$ is connected).

  We now perform two rooted subtree movements, as illustrated in
  Figure \ref{fig:pseudotwostep}. First, break $A_t^m$ into two rooted
  trees $(S,r)$ and $(T,r)$ at $r$, so that $L_{m+1}\in S$ and
  $L_1,\ldots,L_m\in T$; move $S$ and rejoin in the form
  $\join{S}{r}{T}{q}$. Next, break this into the same two rooted
  subtrees $(S,r)=(S,q)$ and $(T,q)$, and move $T$ to rejoin in the
  form $\join{T}{q}{S}{p}$. In this way, the links $L_1,\ldots,L_m$
  are not disturbed, and $L_{m+1}$ is rooted properly, i.e. this new
  hinged pseudofigure is $A_t^{m+1}$. Repeat this procedure to
  finally obtain $A_t^k = A_t$ by rooted subtree moves, as desired.
\end{proof}

\subsection{Rectangle-to-Rectangle Construction}

Here we describe our hinged
dissection between any two equal-area rectangles.
This step is heavily based on the classical (non-hinged)
dissection, modified using our subtree
operations to allow all operations to be hinged.

Take two rectangles of equal area, $A$ and $B$, and
suppose $B$ has the smallest minimum side length.
We begin by aligning both rectangles with their shorter
edge on the horizontal axis
and longer edge on the vertical axis, and identifying
the two top left vertices. We then rotate $B$ counterclockwise
until its lowest vertex is horizontally aligned with the base of $A$.
Label the vertices $a_i$ and $b_i$ for $1 \leq i \leq 4$, starting
at the top left and moving clockwise.

At this point the horizontal cross-section of both rectangles
is equal. Now, cut
$B$ along edge $a_2 a_3$ (the right side), hinging at the bottom of
the cut, and rotate the extended portion of $B$ clockwise
by $\pi$ to cover a strip of $A$. Again, cut $B$, now along
the left side $a_1 a_4$, hinging at the bottom, and rotate it back in,
covering another horizontal strip of $A$. Continue in this
way until the remaining segment of $B$ extending past $A$'s
boundary is no longer enough to cover an entire horizontal
strip (see Figure~\ref{fig:recttorect_snake}).

Now consider the subtriangle $a_1 b_2 a_2$ of $B$. We cut
it horizontally at half its height, and vertically from
$b_2$, and rotate the resulting components out
to form a rectangular cap with the same width as $A$
(see Figure~\ref{fig:recttorect_cap}).

There are now two cases: either the remaining segment of $B$
extends up and left, or down and right. In the fortunate
former case, we can perform the entire transformation using
only classical-style manipulations: swing the extended portion
back into $A$. This will give a ``triangular'' base, similar
to the triangle that was on top of $A$, but offset horizontally
and wrapping through the edge of $A$. We can make this
rectangular as well by a nearly identical transformation:
cut horizontally at its vertical midpoint, and vertically
as shown in Figure~\ref{fig:recttorect_cap} so that the pieces line up
with the border of $A$ when we swing them out.
A quick case analysis shows that this always works.

Now consider the remaining case where the end of $B$ extends
down and to the right. We can directly reduce this to the
previous case: cut $B$ along edge $a_2 a_3$, hinging at
the {\it top}, and use rooted subtree movement to move this
subtree counterclockwise around the figure to line up with
the left side of $A$ (Figure~\ref{fig:recttorect_wrap}). The configuration
of the base is now as
though $B$ had extended up and left and we rotated the
extended piece down
into $A$ as before, the only change being the kite sweep
at the top vertex. Since this kite sweep can easily
be kept above the vertical midpoint of the triangle,
it doesn't interfere with our new cuts, and 
the same capping strategy works
without alteration.

With these steps, $B$ is transformed into a rectangle having the
same area and width (and therefore same height) as $A$.

\subsection{Proof of Pseudopolynomial Bound}~

\begin{proof}[of Theorem~\ref{thm:pseudopoly}]
Since our construction worked independently on each pair
of equal-area triangles, and the equal-area chainification
step itself was pseudopolynomial, it suffices to show that
our transformation is pseudopolynomial when applied to
a single pair of triangles.

We analyze the maximum number of pieces produced by
our dissection by considering the related question of the
smallest possible non-zero value that can be computed
at any intermediate step of the dissection. This value
gives a lower bound on, for instance, the smallest distance
between any two distinct vertices. If this distance is at
least inverse polynomial, then a simple area argument shows
that the number of pieces in the dissection is also at
most polynomial.

Our task, then, is to show that the smallest non-zero
value produced during the computation is indeed inverse
polynomial, a question that can be dealt with almost
entirely algebraically. We will use the bounds from
\cite{mehlhorn}. Under these bounds,
it suffices to show that if we view all numerical computations
as a DAG, then algebraic extensions, multiplications,
and divisions (by previously computed expressions) are never
nested to more than constant depth, and that addition
and subtraction, and multiplication or division by fixed
constants, are never nested to more than linear depth.
Under these constraints, if we are still able to compute
the coordinates of all vertices on our figure, we will be done.
For this section, refer to addition and subtraction, and
multiplication and division by fixed constants, as
simple arithmetic operations.

It is important in maintaining constant depth of
multiplications that we never rotate the same point
or vector more than a constant number of (nested)
times. For instance, while for simplicity we have
described adjacent kites as touching along a common
angle, to attain a pseudopolynomial bound we cannot
do this since each consecutive kite would give another
nested rotation. 
%Thus, for example, in our kite sweeps, adjacent
%kites will not meet along a common edge, since computing
%the sequence of coordinates induced by these adjacent
%angles would involve many sequential rotations.
Instead,
whenever we need to perform a rotation, we will use
a rational approximation of the angle. We can approximate
an angle within an arbitrary error bound $\delta$
by using the ratios of integers that are polynomial in
$1/\delta$ (for instance by solving for an appropriate
Pythagorean triple). We can further ensure that we always
have plenty of room for such approximations by restricting
ourselves to free regions that use only half of the true
available angular space, so we know that there will also
be polynomial error tolerance built into any desired
angles, so any single rotation will still only involve
polynomial values. Kite sweeps will thus be spread out
near each other but with angles ensuring they do not
intersect.

Now let us consider the arithmetic/algebraic depth of the
expressions produced during each step of our transformation.
As mentioned above, for the purposes of analysis we
consider the procedure as beginning with two equal-area
rectangles and then map back to the original triangles
by adding the appropriate pseudocuts.

First: When we rotate the narrower rectangle to have the
same horizontal cross-section as the first, this requires a
quadratic extension, as well as a constant number of
multiplications and divisions to rotate all the vertices of
the second rectangle.

Second: snaking the second rectangle back and forth along
the first (Figure~\ref{fig:recttorect_snake}) requires
only additions and subtractions, since
all rotations are by $\pi$. Since each pass of the rectangle
must cover a strip whose height is at least the width of
the second rectangle (which is at least half the height
of its originating triangle), we require a linear number
of passes,
so the depth of simple arithmetic operations here is at
most linear.

Third: Capping the triangles at the top and bottom
(Figure~\ref{fig:recttorect_cap}). This also
requires only rotations by $\pi$, except for the extra step
of finding the midpoints of the triangles' ascent, which is just a
simple arithmetic operation (division by 2).

Fourth: Moving $B$'s subtree back to the left side of
$A$ (Figure~\ref{fig:recttorect_wrap}), if
necessary. The base vertices of the triangles cut out of the
border of the figure can easily be placed on rational points
with low relative denominator.
We will choose the inner angle of the triangles, as described
earlier, to remain within the free area along the border while
requiring a constant number of multiplications and divisions
with suitable rational numbers. After
choosing the interior angle, finding the interior vertex of a
triangle can be done by intersecting its two edges, which
also takes a constant number of multiplications and divisions,
and can be done independently for each internal vertex to
prevent nesting. Note similarly that each edge traversed by
the path can be handled independently.

This still leaves the kite sweep inside the moved subtree.
We can produce this from the triangle chain:
each triangle pair can be made into a kite by rotating
by $\pi$. Once this is done for each triangle pair,
we place the kites inside the appropriate vertex by
choosing approximate rotation angles as described
above so the kites are non-overlapping. Each rotation requires
a constant number of multiplications per vertex, but again each
kite can be handled independently (and uses only expressions
from its originating triangles, which were themselves computed
independently of other triangle pairs), so we still preserve constant
depth in our computations.

After this movement is complete, the bottom is capped identically to
the previous step.

Fifth: The pseudocuts. We consider these independently. Since
there are only four of them, if we show that any single pseudocut
increases the simple arithmetic depth of an existing figure
at most linearly, and the depth of remaining operations by
at most a constant, we will be done.

Consider a single pseudocut. We know that it intersects with
at most a polynomial number of edges for the simple
inductive reason that so far there are only polynomially many.
Furthermore, it lies in our existing algebraic
extension, since all cuts were simple rational cuts with
respect to the original input triangles, and the transformation
to rectangles involved only rotations by $\pi$ and the single
aligning rotation we performed in the first step. Thus, the
coordinates defining the cut are so far at only constant
computational depth.

Now, the initial vertices we need to add to our figure
are the intersections
of the cut line and any incident edges in the figure.
Each of these can be computed independently of the rest,
and each requires a constant
number of multiplications and divisions. After this we need to
cut out the triangle chains along the boundary of the cut, as
well as the kite sweep at each joint. Once more we exploit the
fact that each kite can be dealt with mostly independently:
the bases of the triangles are simple rational points along
the cut edges. The interior angles can be approximated as before,
though now we may need to compute a polynomial number of them
because of the nested triangles. However, in that case each
interior angle can also be dealt with independently, so even
though there are many interior vertices corresponding to each
base line, they are all independently still at constant
computational depth. The kite sweeps are dealt with exactly
as above, a rotation by $\pi$ followed by an approximate
rotation into the appropriate vertex while avoiding intersections.
As all of these are still independent (relying only on the
constant-depth computation to produce the appropriate base
vertices,
plus the constant-depth computation to produce the matching
interior vertex, plus the constant-depth computation to produce
a suitable rotation angle for the kite), all of this is still
done in only constant depth.

Sixth: Unaltered subtree movement. This follows from
the same argument as the ordinary subtree
movement from the fourth step. The changes in the specific
vertex used for the kite sweep make no difference to the
computational depth required. Since there are only two
such subtree movements, the added simple arithmetic
depth is again linear, and depth of remaining operations
is again constant.

Thus, we see that the computations of all steps together
remain within the required bounds, and thus all vertices
are indeed at least an inverse polynomial distance apart.
\end{proof}
\fi

\end{document}